\documentclass[twoside,leqno]{article}

\usepackage[letterpaper]{geometry}

\usepackage{hyperref}
\usepackage{amsmath}
\usepackage{amssymb}
\usepackage[capitalize, nameinlink]{cleveref}

\newcommand{\qed}{\qquad\vbox{\hrule height0.6pt\hbox{%
   \vrule height1.3ex width0.6pt\hskip0.8ex
   \vrule width0.6pt}\hrule height0.6pt}\outerparskip 0pt}
\usepackage{ltexpprt}

\usepackage[numbers,sort&compress]{natbib}
\usepackage{graphicx}
\usepackage{thmtools,thm-restate}
\usepackage{nicefrac}
\usepackage{xspace}
\usepackage{algpseudocode}
\usepackage{algorithm}
\usepackage{caption}
\usepackage{mathtools}
\usepackage{outlines}

\algdef{SE}[DOWHILE]{Do}{doWhile}{\algorithmicdo}[1]{\algorithmicwhile\ #1}

\newcommand{\cB}{\mathcal{B}}
\newcommand{\LOCAL}{\ensuremath{\mathsf{LOCAL}}\xspace}
\newcommand{\poly}{\operatorname{\text{{\rm poly}}}}
\newcommand{\eps}{\varepsilon}
\newcommand{\lovasz}{Lov\'{a}sz\xspace}
\newcommand{\inter}{\operatorname{inter}}
\newcommand{\intra}{\operatorname{intra}}
\newcommand{\tail}{\operatorname{tail}}
\newcommand{\head}{\operatorname{head}}

\begin{document}

\newcommand*\samethanks[1][\value{footnote}]{\footnotemark[#1]}

\title{\Large On the Locality of Hall's Theorem}
\author{
Sebastian Brandt
\and Yannic Maus \thanks{
This research was funded in whole or in part by the Austrian Science Fund (FWF) \url{https://doi.org/10.55776/P36280}. For open access purposes, the author has applied a CC BY public copyright license to any author-accepted manuscript version arising from this submission.
}
\and Ananth Narayanan
\and Florian Schager \samethanks
\and Jara Uitto}

\date{}

\maketitle

\fancyfoot[R]{\scriptsize{Copyright \textcopyright\ 2025 by SIAM\\
Unauthorized reproduction of this article is prohibited}}

\begin{abstract} \small\baselineskip=9pt

The last five years of research on distributed graph algorithms have seen huge leaps of progress, both regarding algorithmic improvements and impossibility results: new strong lower bounds have emerged for many central problems and exponential improvements over the state of the art have been achieved for the runtimes of many algorithms.
Nevertheless, there are still large gaps between the best known upper and lower bounds for many important problems.

The current lower bound techniques for deterministic algorithms are often tailored to obtaining a logarithmic bound and essentially cannot be used to prove lower bounds beyond $\Omega(\log n)$.
In contrast, the best deterministic upper bounds, usually obtained via network decomposition or rounding approaches, are often polylogarithmic, raising the fundamental question of how to resolve the gap between logarithmic lower and polylogarithmic upper bounds and finally obtain tight bounds.

We develop a novel algorithm design technique aimed at closing this gap.
It ensures a logarithmic runtime by carefully combining local solutions into a globally feasible solution.
In essence, each node finds a carefully chosen local solution in $O(\log n)$ rounds and we guarantee that this solution is consistent with the other nodes' solutions without coordination. The local solutions are based on a distributed version of Hall's theorem that may be of independent interest and motivates the title of this work. 

We showcase our framework by improving on the state of the art for the following fundamental problems: edge coloring, bipartite saturating matchings and hypergraph sinkless orientation (which is a generalization of the well-studied sinkless orientation problem). 
For each of the problems we improve the runtime for general graphs and provide asymptotically optimal algorithms for bounded degree graphs.
In particular, we obtain an asymptotically optimal $O(\log n)$-round algorithm for $(3\Delta/2)$-edge coloring in bounded degree graphs. The previously best bound for the problem was $O(\log^4 n)$ rounds, obtained by plugging in the state-of-the-art maximal independent set algorithm from [Ghaffari, Grunau, SODA'23] into the $3\Delta/2$-edge coloring algorithm from [Ghaffari, Kuhn, Maus, Uitto, STOC'18]. 
\end{abstract}

\section{Introduction}
In recent years, the area of distributed graph algorithms has undergone an incredible development with faster and faster algorithms for classic local graph problems, general derandomization methods, and breakthrough results for proving lower bounds.
Nevertheless, apart from highly artificial problems and problems that trivially admit a constant-time algorithm or require linear time, there is almost no local graph problem for which matching upper and lower bounds on the distributed complexity are known. 
For example, the general derandomization method for local graph problems \cite{BE11,GKM17,GHK18,RG20} yields polylogarithmic-time deterministic distributed algorithms, while the best known lower bounds for these problems are at most logarithmic.
Recently, highly optimized algorithms have been developed for problems like the maximal independent set problem or the intensively studied $(\Delta+1)$-coloring problem for graphs with maximum degree $\Delta$ \cite{FGKR23,Ghaffari-rounding-2021,GG23}.
Yet, they seem to be unable to close the polynomial gap to the (at best) logarithmic lower bounds.
(For $(\Delta+1)$-coloring the gap is much larger, as the best known lower bound of $\Omega(\log^*n)$ still comes from Linial's seminal work~\cite{linial92}.) 
To close the gap from the lower bound side seems even harder: as essentially all lower bound techniques \cite{linial92,KMW16,Brandt19,Brandt22marks} only work on high-girth graphs, it is currently out of reach to prove genuine superlogarithmic lower bounds (except for highly artificial or global problems).

In conclusion, there is a need for new algorithmic techniques for closing the gap between upper and lower bounds.
In this work, we address this need by providing a \emph{new algorithm design technique} that gives rise to deterministic logarithmic-time algorithms for local problems.
This leads to improvements for a number of problems; for instance, we obtain improved algorithms for edge coloring with few colors (that are tight on constant-degree graphs). 

Before explaining our new technique in more detail, we introduce the model of computation.

\paragraph{Model of computation.}

The model of computation we study is the classic \LOCAL model of distributed computation~\cite{linial92} (and its generalization to hypergraphs).
In the \LOCAL model, a communication network is abstracted as an $n$-node (simple) graph, where nodes are computational entities equipped with a unique ID and edges serve as communication channels. Communication happens in synchronous rounds, in each of which a node can perform some arbitrary local computations, and send one message to each of its neighbors in the graph. 
The time complexity of an algorithm is the number of rounds until each node has output a solution, e.g., the orientation of each of its incident edges. A hypergraph $G=(V,E)$ can be modeled as a bipartite graph $\cB_G$ where the nodes of $G$ form one side of the bipartition of the nodes of $\cB_G$ (which we will call the \emph{vertex side}) and the hyperedges of $G$ form the other side of the bipartition (which we will call the \emph{hyperedge side}).
There is an edge between a node of $\cB_G$ that corresponds to a node $v$ of $G$ and a node of $\cB_G$ that corresponds to a hyperedge $e$ of $G$ if and only if $v \in e$.
The classic setting for ``\LOCAL on a hypergraph $G$'' (that we also use) is the standard \LOCAL model on $\cB_G$.  

\paragraph{A new technique.}
In the following, we provide a high-level overview of our new technique.
For a more extensive overview, see \Cref{sec:dihath}.
Informally, the basic idea of our technique is as follows.
First, each node of the network computes a local solution for a subgraph in which it is contained. Then, the different solutions produced by all nodes of the network are combined to a global solution for the whole graph. In contrast to most other algorithms, there is essentially no additional coordination when combining solutions, except that each node should know its own output in all local solutions in which it appears.
The way in which we find the subgraphs on which the local solutions are to be computed is based on carefully carving out subgraphs whose removal does not place any burden on the solvability of the problem on the remaining graph.
In the following, we illustrate this approach in the context of matching problems, in which the implementation of our technique is based on a distributed version of Hall's theorem that we prove.
We note that in our work, we will make the outlined approach work directly only for matching and related orientation problems, but that the results obtained thereby will then enable us to prove bounds for further problems, such as edge coloring or splitting problems.

\paragraph{Distributed Hall's Theorem.}
The following classic result by Hall provides a characterization for the existence of a matching in a bipartite graph that saturates all nodes of one side.
\begin{figure}
    \centering
    \includegraphics[height=4cm]{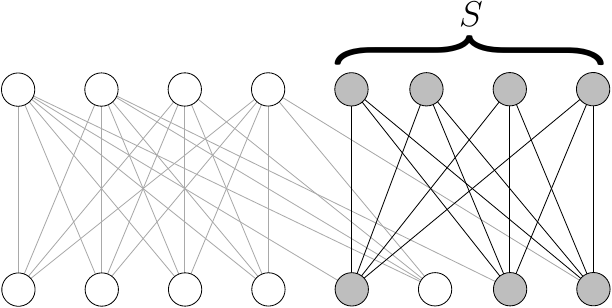}
    \caption{The bipartite representation of a hypergraph with minimum degree $\delta = 3$ and rank $r = 4$ for each edge. The set of nodes $S$ and their neighbors are illustrated as grey nodes. The set $S$ violates ``Hall's condition'' and hence, cannot be perfectly matched.}
    \label{fig: non-solvable}
\end{figure}

\smallskip

\textsc{Hall's Theorem~\cite{Halls}.}
\emph{A bipartite graph with node sets $V$ and $U$ has a $U$-saturating matching if and only if $|N(S)|\geq |S|$ for all $S\subseteq U$.}

\smallskip

In the context of saturating matchings, the subgraphs that ``can be removed without creating problems'' mentioned in the above outline of our approach can be specified as follows: they are subgraphs that allow to find a saturating matching inside the subgraph such that no ``remaining'' node (on the side to be saturated) ``loses'' any potential matching partner.
We call such a subgraph a \emph{Hall graph}.
Now we can rephrase our approach as being based on a local version of Hall's theorem: as our main technical contribution, we show that every node (of a multihypergraph\footnote{Formally, we will consider bipartite graphs in the hypergraph formalism (explained below). We note that many of our results work for multihypergraphs which are hypergraphs in which the same edge may appear multiple times.}) is contained in a small-diameter Hall graph. 

\begin{restatable}[Distributed Hall's Theorem]{theorem}{thmDistHall}
\label{thm:distributedHall}
 Each node in any $n$-node multihypergraph with minimum degree $\delta$ and maximum rank $r<\delta$ is contained in a Hall graph with diameter $\log_{\nicefrac{\delta}{r}}n$.
\end{restatable}

Since, in the \LOCAL model (and its generalization to hypergraphs), each node can collect its entire $T$-hop neighborhood in $T$ rounds, the upper bound on the diameter of the Hall graphs in \Cref{thm:distributedHall} enables the nodes to compute the aforementioned local solutions in $O(\log_{\nicefrac{\delta}{r}}n)$ rounds, ultimately giving rise to algorithms with logarithmic runtime. In several cases this closes the gap to the respective lower bound.

\medskip

\paragraph{Hypergraph Sinkless Orientation.}
While the problem of finding a saturating matching is a fundamental algorithmic graph problem in its own right, there is a second reason for the importance of this problem in a distributed context: as we discuss in \Cref{sec:hsoexplained}, it is equivalent to the hypergraph generalization of a graph problem that has proved to be crucial for the development of various fundamental lines of research of recent years---sinkless orientation~\cite{LLL_lowerbound}. The objective of the \emph{(hypergraph) sinkless orientation problem ((H)SO)} is to orient the (hyper)edges of a (hyper)graph such that every vertex has at least one outgoing (hyper)edge, where a hyperedge is outgoing for exactly one of its incident nodes and incoming for all others.

For an overview of the vital role that sinkless orientation has played in the development of distributed graph algorithms in the last decade, we refer the reader to \Cref{sec:furtherwork}.
We expect an understanding of HSO to be similarly essential for understanding \emph{distributed computation on hypergraphs}, which has been a topic of substantial interest~\cite{Harris20,GHK18,AHN23,BBKO23,KZ18,GKMU18,Fischer17hyper}.
Moreover, similarly to how sinkless orientation is a highly useful subroutine, e.g., for splitting problems\footnote{In fact, we also use the sinkless orientation problem as a subroutine in one part of our edge coloring algorithm. } \cite{GS17,Splitting20}, we expect HSO to be an essential ingredient for solving other problems (both on graphs and hypergraphs), emphasizing the importance of finding \emph{optimal} algorithms for HSO.
We remark that we also provide concrete evidence for the usefulness of HSO as a subroutine by showing that we can use HSO to obtain improved algorithms for edge coloring and further problems.

\subsection{Our Contributions: Main results}

\subsubsection{Edge Coloring}
	
As the main application of our technique we prove the following theorem. 

\begin{restatable}{theorem}{thmedgeColoring} \label{thm:edgeColoring}
There is a deterministic $O(\Delta^2 \cdot \log n)$-round \LOCAL algorithm that computes a $3\Delta/2$-edge coloring on any $n$-node graph with maximum degree $\Delta$.
\end{restatable}
The main strength of this theorem is its dependence on the number of nodes in the network. In fact, we obtain the following corollary for edge coloring constant-degree graphs. 
\begin{corollary}
	\label{cor:EdgeColoring}
There is an $O(\log n)$-round \LOCAL algorithm that computes a $3\Delta/2$-edge coloring on any $n$-node graph with constant maximum degree $\Delta$.
\end{corollary}
The runtime in \Cref{cor:EdgeColoring} matches the lower bound given by \cite{CHLPU18} that holds for computing any edge coloring with fewer than $2\Delta-1$ colors. In the centralized setting a coloring with $2\Delta-1$ colors can be computed via a simple greedy algorithm, and in the distributed setting such colorings can be computed either in $O(\poly \log \Delta +\log^*n)$ \cite{Balliu22-edge-coloring},  in $O(\log^2\Delta\log n)$ rounds \cite{Ghaffari-rounding-2021}, or in $O(\log^2 n\poly\log\log n)$ rounds \cite{GG23}.  To the best of our knowledge \Cref{cor:EdgeColoring} presents the first classic graph coloring problem with a provably logarithmic complexity via asymptotically matching upper and lower bounds. In general, there is a very small list of problems with asymptotically matching upper and lower bounds in this runtime regime, with the sinkless orientation problem being the most prominent example \cite{LLL_lowerbound,ghaffari17}. There are many works studying deterministic algorithms for edge coloring with fewer than $2\Delta-1$ colors, e.g.,~\cite{CHLPU18, bernshteyn2024fastalgorithmsvizingstheorem, BERNSHTEYN2022319, HMN22,Davies23, GKMU18, Su-Vizing2019}. 
In $O(\log^5 n)$ rounds\footnote{The constant in the $O$-notation hides a $\Delta^{84}$ term, which could likely be improved but not to a linear dependency.}, we know how to deterministically obtain an edge coloring with $\Delta+1$ colors colors on constant-degree graphs~\cite{bernshteyn2024fastalgorithmsvizingstheorem, BERNSHTEYN2022319}. This matches the existential bound proven by Vizing in the 60s of the last century \cite{vizing1964}.
Relaxing the number of colors a bit to $\Delta + 2$, we know how obtain a $\poly \Delta \cdot \log^3 n$ randomized algorithm \cite{Su-Vizing2019}.
While these algorithms still run in $\poly(\Delta,\log n)$ rounds their dependency on the number of nodes in the network is quite far from the logarithmic lower bound, given by \cite{CHLPU18}. 
Even algorithms tailored for more than $\Delta + 1$ colors have not been able to get close to the lower bound. 
For example, the deterministic $3\Delta/2$-edge coloring algorithm in~\cite{GKMU18} has runtime $\widetilde{O}(\Delta^3\log^4n)$. The original paper states a much slower runtime, but it improves by using the recent state-of-the-art maximal independent set algorithm \cite{GG23} for solving the hypergraph matching problems that appear as subroutines in their algorithm, see \Cref{sec:nutshell} for details.  Again, using more than $\Delta+2$ colors, recent progress has provided randomized  edge coloring algorithms running in $\poly\log\log n$-round algorithms for $(1+o(1))\Delta$-edge coloring for graphs with a sufficiently large maximum degree \cite{HMN22,Davies23}. 

By employing a standard divide-and-conquer strategy we can first decompose the graph into small degree graphs, each of which we color with its own set of colors using \Cref{thm:edgeColoring}. We obtain the following corollary.  
\begin{restatable}[$(3/2 + \varepsilon)\Delta$-edge coloring, deterministic]{corollary}{corEdgeColoring}
\label{cor:EdgeColoringFast} For any $\eps>1/\Delta$ there is a deterministic $O(\eps^{-2}\log^2 \Delta \cdot \poly \log \log \Delta \cdot \log n)$-round \LOCAL algorithm that computes a $(3/2  + \varepsilon)\Delta$-edge coloring on any $n$-node graph with maximum degree $\Delta$.
\end{restatable}
The actual (more involved) runtime in our proof is slightly better than the one stated in this corollary. For constant $\eps$ it nearly matches the state of the art (with $n$-dependency limited to $\log n$) for the easier greedily solvable $(2\Delta-1)$-edge coloring algorithm \cite{Ghaffari-rounding-2021}. The difference is only a $\log\log \Delta\cdot \log^{1.71}\log\log \Delta$ factor. \Cref{cor:EdgeColoringFast} is significantly faster than all prior algorithms for coloring with $(3/2+\eps)\Delta$ colors. 

See \Cref{table:EdgeColoring} for a comparison of our edge coloring results with various prior algorithms, in each of which we have updated subroutines with state-of-the-art algorithms. 
 
\subsubsection{Hypergraph Sinkless Orientation (HSO)}
\label{sec:hsoexplained}

In this section, we discuss our results for computing hypergraph sinkless orientations, or equivalently node saturating matchings in bipartite graphs. 
Let us first explain this equivalence.  Recall the way in which a hypergraph $G$ can be modeled as a bipartite graph $\cB_G$, explained in the model description.
Selecting a single outgoing edge for each node in a solution for HSO on some hypergraph $G$ gives rise to a maximum matching in the bipartite representation $\cB_G$ of the hypergraph that saturates all nodes of $\cB_G$ corresponding to nodes in $G$.
Conversely, any such matching gives rise to a solution for HSO.
As such, the two problems are equivalent, and in particular, the two problems have the same asymptotic complexity.

\paragraph{When do solutions for HSO exist (Hall's Theorem)?} 
As for the sinkless orientation problem on graphs, the HSO problem requires some constraints on the input instances to avoid nonexistence of a solution.
Let $\delta$ denote the minimum number of hyperedges incident to any node in $G$, and $r=\max_{e\in E} |e|$ the maximum rank of a hyperedge.
In other words, $\delta$ is the minimum degree of the nodes on the vertex side of $\cB_G$ and $r$ is the maximum degree of the nodes on the hyperedge side of $\cB_G$.
Moreover, for a set $S$ of nodes in a graph $G = (V,E)$, let $N(S)$ denote the set of all nodes $u \in V$ such that there exists an edge $\{ u, v \} \in E$ with $v \in S$.

If $\delta\geq r$, then \emph{Hall's Theorem} implies that a matching saturating all nodes on the vertex side exists.
Moreover, already if we weaken this condition slightly to $\delta \geq r - 1$, there are graphs for which no such matching exists (see \Cref{fig: non-solvable}).
Hence, the problems of HSO and bipartite saturating matching necessarily require $\delta\geq r$.
Moreover, while existence is guaranteed if $\delta = r$, we show that the two problems actually require $\delta>r$ if we want to achieve sublinear complexity.
\begin{restatable}[HSO linear lower bound]{theorem}{thmHSOLowerBoundGlobal}
\label{thm:linearLowerBound}
For any fixed $\delta$, there is no sublinear deterministic algorithm to compute an HSO on all hypergraphs with minimum degree $\delta$ and maximum rank $r=\delta$. 
\end{restatable}
The lower bound is easy to see in the special case of $\delta=r=2$ where the problem corresponds to computing a consistent orientation on a cycle, but more difficult to obtain for larger values of $\delta = r$.
 To summarize, in case of $\delta<r$, there may be no solution to HSO,  in case of equality $\delta=r$, a solution to HSO exists, but computing one requires $\Omega(n)$ rounds. For all other cases we prove the following theorem.
 
\begin{restatable}{theorem}{thmHSO}
\label{thm:hypergraphSO}
There is a deterministic $O(\log_{\frac{\delta}{r}}n)$-round algorithm for computing an HSO in any $n$-node multihypergraph with minimum degree $\delta$ and maximum rank $r<\delta$. 
\end{restatable}
For $r=2$, HSO equals the (graph) sinkless orientation problem.  In that case the condition $\delta>r$ equals the standard assumption for SO that the minimum degree of the input graph is $3$. Similarly, as for HSO, otherwise one requires either linear time or the problem does not even have a solution. One can extend \Cref{thm:hypergraphSO} to general hypergraphs without this condition if all nodes whose degree is at most the maximum rank do not need to have an outgoing hyperedge.

By a combination\footnote{More precisely, the lower bound follows from the fact that HSO is a so-called \emph{fixed point} in the round elimination framework introduced in~\cite{Brandt19}, which implies a deterministic lower bound of $\Omega(\log_{\delta \cdot r} n) = \Omega(\log_{\delta} n)$ rounds for deterministic algorithms and $\Omega(\log_{\delta \cdot r} \log n) = \Omega(\log_{\delta} \log n)$ rounds for randomized algorithms (as shown in, e.g., \cite[Section 7, arXiv version]{BBKO2021hideandseek}). The fact that HSO is a fixed point is a straightforward extension of~\cite[Section 4.4, arXiv version]{Brandt19} to hypergraphs.} of previous results~\cite{Brandt19,BBKO2021hideandseek}, solving HSO requires $\Omega(\log_{\delta} n)$ rounds deterministically (and $\Omega(\log_{\delta} \log n)$ rounds randomized).
Hence, our deterministic algorithm is provably asymptotically optimal whenever $\delta \geq r^{1 + \eps}$ for some (arbitrarily small) positive constant $\eps$ (which, for instance, is always satisfied in the case that $r$ is constant).
But also in the case that our upper bound does not match the lower bound, only a small factor of $(\log \delta) / (\log \nicefrac{\delta}{r})$ remains between upper and lower bound.
Moreover, if $\delta \geq (1 + \eps)r$ for some (arbitrarily small) positive constant $\eps$, then our algorithm has logarithmic runtime (which, again, holds amongst others in the case that $r$ is constant).

\subsubsection{Implications}

\paragraph{Bipartite Maximum Matching.}
Via the interpretation of an HSO on a hypergraph $G$ as a matching in the graph's bipartite representation $\cB_G$ we obtain the following corollary. 
It is a direct consequence of \Cref{thm:hypergraphSO}. 

\begin{restatable}[Maximum Matching, deterministic]{corollary}{corMatchingDeterministic}
	\label{cor:maximumMatchingDeterministic}
	There is a deterministic $O(\log_{\frac{\delta}{r}}n)$-round algorithm for computing a left side saturating (and therefore maximum) matching in any $n$-node bipartite graph with minimum left side degree $\delta$ and maximum right side degree $r< \delta$.
\end{restatable}
The \emph{maximal} matching problem admits a deterministic $O(\log^*n)$-round algorithm~\cite{panconesi-rizzi} on \emph{constant-degree} graphs and a recent breakthrough work shows that on general graphs there is no $o(\Delta + \log n / \log \log n)$-time deterministic and no $o(\Delta + \log \log n / \log \log \log n)$-time randomized algorithm~\cite{BBHORS19MMlowerBound}.
Despite significant progress decreasing the deterministic runtime from $O(\log^7 n)$~\cite{Hanckowiak1998} to $O(\log^4 n)$~\cite{hanckowiak01} to $O(\log^3 n)$~\cite{Fischer17}, the best runtime (as a function of $n$) of $O(\log^2n \poly\log\log n)$~\cite{GG23} is still substantially larger than logarithmic. 
Our result provides a deterministic $O(\log n)$-round maximal matching algorithm for a large class of bipartite graphs with unbounded degrees. 

We remark that the techniques from~\cite{BBHORS19MMlowerBound} imply a lower bound of $\Omega(\min\{r, \log_{\delta} n\})$ rounds deterministically (and $\Omega(\min\{r, \log_{\delta} \log n\})$ rounds randomized) for maximal matching in our setting.
Hence, even for the fundamental maximal matching problem our deterministic runtime is asymptotically optimal for a certain parameter range, including the regime where $r \geq (\log n)/(\log \log n)$ and $\delta \geq r^{1 + \eps}$ for some constant $\eps > 0$.

\paragraph{Weak Splitting.}
The objective of the \emph{weak splitting} problem~\cite{GKM17} is to color the right-hand side of a bipartite graph with two colors such that every node on the left-hand side has at least one neighbor with each color.
While seemingly easy, it turns out that this problem admits no deterministic or randomized algorithms with runtime $o(\log n)$ and $o(\log \log n)$~\cite{FGLLL17, LLL_lowerbound, BGKMU19, Chang2016}. We obtain the following result for weak splitting by a reduction to the HSO problem.
This improves on the prior $O(\log^3 n \poly\log\log n)$-round algorithm~\cite{BGKMU19}.

\begin{restatable}[Weak splitting]{corollary}{corXYSplitting}\label{cor:splitting}
	There is a deterministic $O(\log_{\frac{\delta}{r}} n)$-round algorithm for the weak splitting problem where $r$ is the maximum degree on the right-hand side of the bipartite graph and $\delta\geq 2(r + 1)$ is the minimum degree on its left-hand side.
\end{restatable}

\paragraph{Implications for randomized algorithms.}
To the best of our knowledge, the HSO problem has not been studied explicitly, but one can derive algorithms for the problem by modeling it as an instance of the Constructive \lovasz Local Lemma (LLL) \cite{CPS17,FGLLL17,RG20}; see \Cref{sec:randomized} for details.  Using known algorithms it can either be solved in  $O(\poly(\delta,r)+\poly\log\log n)$~\cite{FGLLL17,RG20} or in $O(\log n)$ rounds~\cite{CPS17}. 
Actually, it is conjectured that on constant-degree graphs, there are randomized algorithms for LLL that run in $O(\log\log n)$ rounds~\cite{Changhierarchy19}. 
In the special case of trees, this conjecture has been verified~\cite{CHLPU20}.
But, to the best of our knowledge, on general (bounded-degree) graphs a runtime of $O(\log\log n)$ rounds has only been achieved for the special LLL-type problem of sinkless orientation. 
We add HSO-type problems to that list.
More generally, we obtain the following result that is exponentially faster than the previous algorithms for most choices of $\delta$ and $r$.

\begin{restatable}[HSO randomized upper bounds]{theorem}{thmHSORandomized} 
\label{thm:randomizedHSO}
	 There is a randomized algorithm that w.h.p.\ computes an HSO on any hypergraph with rank $r$ and minimum degree $\delta\geq 320 r\log r$ with runtime 
  $$
  O\left(\log_{\frac{\delta}{r}}\delta+  \log_{\frac{\delta}{r}} \log n\right).$$
    If $r\geq 100 \log n$, an alternative algorithm solves the problem in $O(\log\log n/\log\log \log n)$ rounds.   
\end{restatable}
In fact, if $\delta\leq \poly(\log n)$ the runtime of \Cref{thm:randomizedHSO} becomes $O(\log_{\frac{\delta}{r}}\log n)$, which is constant when $\delta/r=\Omega(\log^{\eps} n)$, for some constant $\eps>0$. 
This is perhaps surprising, as for $\eps<1$, randomly orienting each hyperedge does not solve HSO with high probability, but still, we obtain a constant-time algorithm.
Moreover, still if $\delta\leq \poly(\log n)$, due to the aforementioned randomized lower bound of $\Omega(\log_{\delta} \log n)$ rounds we obtain in general that, again, only a small factor of $(\log \delta) / (\log \nicefrac{\delta}{r})$ remains between upper and lower bound; in particular our algorithm is asymptotically optimal when $\delta \geq r^{1 + \eps}$ for some (arbitrarily small) positive constant $\eps$. 

We also obtain a randomized algorithm for maximum matching in bipartite graphs. 
We refer to \Cref{sec:priorWorkComparison} for a detailed comparison to prior algorithms for bipartite matching problems. 
\begin{restatable}[Maximum Matching, randomized]{corollary}{corMatchingRandomized}
\label{cor:maximumMatchingRandomized}
 There is a randomized algorithm that, w.h.p.\, computes a left side saturating matching in any $n$-node bipartite graph with maximum right side degree $r$ and minimum left side degree $\delta\geq 80 r\log r$ in $O(\log_{\frac{\delta}{r}}\delta+\log_{\frac{\delta}{r}}\log n)$ rounds. 
  
  If $r\geq 100\log n$, there is an algorithm to solve the problem w.h.p.\ in $O(\log\log n/\log\log\log n)$ time. 
\end{restatable}

\subsection{Our Technique in a Nutshell }
\label{sec:nutshell}
\subsubsection{Edge Coloring}\label{sec:nutedge}
Our $3\Delta/2$-edge coloring algorithm is based on the framework provided in~\cite{GKMU18}.
This framework relies on two crucial subroutines: first the input graph $G$ is partitioned into $\approx \Delta/2$ so-called $(3)$-graphs, and then each of these $(3)$-graphs is edge-colored with a separate set of $3$ colors. More precisely, \cite{GKMU18} iteratively extracts (3)-graphs in a way that reduces the maximum degree of the remaining graph by at least two in each iteration, so that using $3$ fresh colors for each extracted $(3)$-graph results in a $3\Delta/2$-edge coloring of $G$.

We improve the runtime of both subroutines.
The core runtime contribution in the aforementioned extraction procedure comes from the computation of a maximum matching in certain bipartite graphs that fit the framework of \Cref{cor:maximumMatchingDeterministic}. Using our maximum matching algorithm, we can extract a single $(3)$-graph in $O(\Delta \log n)$ rounds. In their work the respective maximum matching problems are solved in $O(\Delta^4\cdot \log^5\Delta \cdot \log^5n \cdot \poly\log\log n)$ rounds, where we already used \cite{Harris19} for improving a subroutine for  hypergraph matching problems (HMs) in their algorithm. The runtime in the original paper~\cite{GKMU18} was significantly slower. Alternatively, the HMs  can be solved via the state-of-the-art maximal independent set algorithm algorithm from \cite{GG23} and obtain a runtime of $O(\Delta^2\log^4n\cdot \poly\log\log n)$ for extracting a single (3)-graph. 

For improving the second subroutine, we develop a novel way of $3$-edge-coloring $(3)$-graphs (outlined below). A $(3)$-graph is a graph with maximum degree $3$ in which nodes of degree $3$ form an independent set. As the line graph of a $(3)$-graph is a graph with maximum degree $3$, Brooks' Theorem implies that it can be colored with $3$ colors, and the state-of-the-art distributed implementation of Brooks' theorem yields a complexity of $O(\log^2 n)$ rounds for $(3)$-edge-coloring $(3)$-graphs, which our new approach improves to $O(\log n)$ rounds. 

\paragraph{Our approach for $3$-edge coloring $(3)$-graphs.}
First, we use a ruling set algorithm to compute a clustering of the input $(3)$-graph that guarantees that each cluster is of constant diameter but at the same time has a sufficiently large number of neighbors in case the cluster is a tree where every edge has $3$ adjacent edges.
As $(3)$-graphs have maximum degree $3$ we can also ensure that each cluster is adjacent to at most a constant number of other clusters and the clustering can be computed in $O(\log^*n)$ rounds.
We will first color all intercluster edges (i.e., edges whose endpoints lie in different clusters) and then all intracluster edges.

Call a cluster \emph{easy} if it contains an edge that has at most $2$ adjacent edges or a cycle.
We show that for easy clusters, any adversarially chosen coloring of the edges outside of the cluster including the adjacent intercluster edges can be completed to a valid $3$-edge coloring of the cluster's edges.
Now consider clusters that are not easy.
As guaranteed by our clustering algorithm, these clusters have many adjacent intercluster edges.
We show that also for any such cluster, any $3$-edge coloring of the edges outside of the cluster can be completed inside the cluster as long as the cluster can decide on the colors of three incident intercluster edges to a certain extent.
This gives rise to our overall approach outlined in the following.

After computing the clustering, we compute a sinkless orientation on the cluster graph obtained by contracting clusters, which ensures that any non-easy cluster has (at least) three outgoing intercluster edges.
Each non-easy cluster chooses three such edges; we say that a cluster \emph{owns} these edges.
Then we compute a helper $3$-coloring of the intercluster edges that we subsequently use to find the final color for each intercluster edge by iterating through the color classes of the helper coloring.
When coloring an edge owned by a cluster, the cluster can decide how to color the edge (under the constraint that the resulting partial coloring is still proper).
As the final step, each cluster simply collects the colors of its adjacent intercluster edges and then completes the coloring inside the cluster.

The runtime of the overall algorithm is dominated by the time it takes to compute the sinkless orientation, which takes $O(\log n)$ rounds; all other builing blocks can be performed in $O(\log^* n)$ rounds or constant time.

\medskip

\subsubsection{Distributed Hall's Theorem}
\label{sec:dihath}
The deterministic algorithm from previous work to find a (non-hypergraph) sinkless orientation makes use of the following simple observation: if you consistently orient a cycle, each node on the cycle will have an outgoing edge (hence, become happy)~\cite{GS17}.
Even more importantly, orienting a cycle cannot make the problem \emph{harder} for any node that is not on the oriented cycle.
In fact, the problem instance can only become \emph{easier}: any adjacent node can now orient its edge toward the cycle and become happy.

Unfortunately, we do not have these properties in the case of hypergraph sinkless orientation.
In hypergraphs, the set of (hyper)edges in a cycle of nodes might contain also nodes that are \emph{not} part of the cycle. See \Cref{fig:sinklessNotWorking} for an illustration.
The fundamental issue is that orienting a cycle in a hypergraph might make the problem instance harder for the remaining graph.
Hence, a more careful approach is needed.

\begin{figure}
    \centering
    \includegraphics[width=0.5\textwidth]{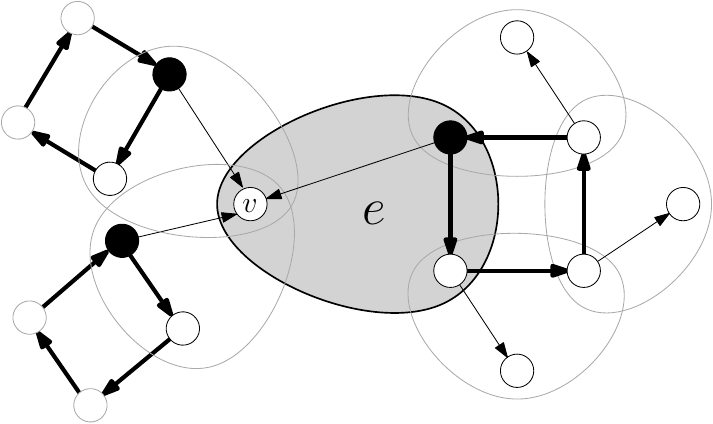}
    \caption{The figure illustrates that consistently orienting cycles can leave an instance unsolvable for node $v$. Each hyperedge is illustrated as an egg-shape over the nodes and the tail of a hyperedge is illustrated through outgoing directed edges. The tails of hyperedges that contain $v$ are drawn black. On the right side, we have a cycle of nodes induced by $4$ hyperedges, shown with bold arrows. Hyperedge $e$ is oriented without making $v$ happy. The cycles on the left can be oriented similarly, leaving $v$ without any outgoing hyperedge.}
    \label{fig:sinklessNotWorking}
\end{figure}

This motivates the definition of a Hall graph. Formally, a Hall graph $H$ is a hypergraph that admits a solution to HSO, or in other words, a graph whose bipartite representation has a matching saturating all nodes on the vertex side.  
Recall our main technical contribution:

\thmDistHall*

Informally, given \Cref{thm:distributedHall}, we can solve HSO on a hypergraph with minimum degree $\delta$ and maximum rank $r<\delta$ in $O(\log_{\frac{\delta}{r}} n)$ rounds as follows.
Each node $v$ collects its $O(\log_{\frac{\delta}{r}} n)$-hop neighborhood and uses the collected information to compute a Hall graph it is contained in.
Moreover $v$ determines all Hall graphs chosen by other nodes that contain $v$.
Then all nodes simulate (part of) a sequential process.
This process consists in going through all computed Hall graphs in some fixed order and orienting the hyperedges of each Hall graph according to some fixed HSO inside the Hall graph.
If a hyperedge is contained in more than one such Hall graph, then its later orientations will overwrite earlier orientations.
We will see that each node has sufficient information to simulate the part of the sequential process relevant for the orientations of its incident hyperedges.
On a high level, the correctness of the algorithm follows from the fact that for each node $v$, the Hall graph containing $v$ processed last will provide $v$ with an outgoing hyperedge whose orientation will not be overwritten afterwards.

The proof of  \Cref{thm:distributedHall} uses the following two lemmas that we restate here in a simplified manner and provide high-level ideas for their proofs.

\smallskip

\textsc{\Cref{lemma:manyEdges} .}  
\emph{Any vertex is contained in a small diameter subgraph with more edges than vertices.}

\smallskip

\smallskip

\textsc{\Cref{lem:nonEmptyHallExists} .}  
\emph{Any graph with more edges than vertices contains a non-empty Hall subgraph.}

\smallskip

\paragraph{Proof ideas for \Cref{lemma:manyEdges,lem:nonEmptyHallExists}}\Cref{lemma:manyEdges} is proven via a ball growing argument and is the only lemma that explicitly uses $\delta>r$ to bound the diameter of these subgraphs. The proof of \Cref{lem:nonEmptyHallExists} is more involved. If the input graph to the lemma is not a Hall graph, it contains a subset of vertices that violates the Hall condition, i.e., its number of incident edges is smaller than the number of vertices in the set. If we remove these vertices and all of its incident edges from the graph, we can again ask whether the resulting graph is a Hall graph, if not, we repeat the argument. We show that, as we start with more edges with vertices, this process has to terminate with a non-empty Hall graph. 

\paragraph{Proof roadmap of \Cref{thm:distributedHall}}We sketch the roadmap for the proof of \Cref{thm:distributedHall} given both lemmas. Also see \Cref{alg:DistributedHall} for an informal ``pseudocode'' for what we are about to sketch. First, we use \Cref{lemma:manyEdges} to find a small diameter subgraph $G'_0$ around node $v$ that contains more edges than nodes. Then, we use  \Cref{lem:nonEmptyHallExists} to show that $G'_0$ (using $|E(G'_0)|\geq |V(G'_0)|$) has to contain a non-empty Hall graph $H_0$. The issue is that $H_0$ may not contain $v$ itself. Hence, we remove $H_0$ from $G$ and repeat the whole process, carving out further graphs $H_1, H_2, \ldots, H_k$ until, as we prove, the last one eventually has to contain node $v$. 
Due to a technicality (edges of $H_i$ may not be edges of $G$ as the edges may have lost vertices), none of these may be Hall graphs in the original graph $H$, but we can ``lift'' them back to $G$ and then return $H_0\cup H_1\cup\ldots\cup H_k$ (informal notation) as the resulting small diameter Hall graph containing $v$. 

\subsection{Comparison to Related Work}
\label{sec:priorWorkComparison}

\paragraph{Splitting problems} \cite{BGKMU19} studies so-called splitting problems in which the objective is to color the hyperedges in a hypergraph with two colors such that every vertex has one incident hyperedge of each color. One of their algorithms can be adapted to work for HSO, as it merely provides a degree rank reduction method (by carefully removing vertices from hyperedges) until every vertex is adjacent to exactly two hyperedges and no other vertex shares these hyperedges. In their algorithm, the vertices then color these two hyperedges with the two desired colors. Instead, one could simply orient these hyperedges outwards and hence solve HSO.
The asymptotics of their procedure do not improve if one only aims at one remaining hyperedge per vertex. The condition under which the algorithm works is $\delta\geq 6r$ and the runtime of the deterministic algorithm is $O(\log^3 n \poly\log\log n)$.
Our algorithm solves the problem for significantly more difficult combinations of $\delta$ and $r$ and it is faster for every setting of parameters, as for $\delta\geq 6r$ our algorithm runs in $O(\log n)$ time deterministically.

\smallskip
\paragraph{Bipartite Matching.} Surprisingly also the left side saturating matching under the condition $\delta>r$ has been studied before \cite{GKMU18}. It appears as a subroutine in the internals of a distributed algorithm for  $3\Delta/2$-edge coloring simple graphs ($\neq$ hypergraphs). 
The main technical contribution of that algorithm is a result showing that any matching that does not saturate the vertices does have an augmenting path of length $\ell=O(d\log n)$. This yields an algorithm for the maximum matching problem as one can iteratively find a maximal independent set (MIS) of such augmenting paths and apply these. It is well known that such an augmentation increases the length of the shortest existing augmenting path and hence after $\ell$ augmentations the matching has to be a maximum matching.  The MIS of augmenting paths can be modeled as a hypergraph MIS, which at the time could be solved in $\poly(d,\log n)$ rounds with a relatively large exponent in the runtime. Despite recent drastic improvements in that runtime, this subroutine still requires non-negligible polylogarithmic runtime. Even if one could obtain the hypergraph MIS for free in each iteration, this approach inherently requires $\Omega(d^2\log^2 n)$ rounds, rendering our method clearly superior.   

It is known that one requires $\Omega(\sqrt{\log n / \log \log n})$ rounds to find a constant (or even polylogarithmic) approximation to fractional maximum matching~\cite{nearsighted}. 
While this does not directly hold for bipartite graphs, there is a simple reduction. 
Take a bipartite double-cover of the input graph and find a fractional matching. Halving the fractional value on each edge loses only a factor of 2 in the approximation. Then, take the resulting value on each edge as the fractional value on the original graph.
Hence, we have the $\Omega(\sqrt{\log n / \log \log n})$ lower bound also for constant-approximate fractional matching on bipartite graphs.
This lower bound directly carries over to the maximal matching problem.
For maximal matching, we know that it cannot be solved in $o(\Delta + \log \log n / \log \log \log n)$-time and in $o(\Delta + \log n / \log \log n)$ randomized and deterministic, respectively~\cite{BBHORS19MMlowerBound}.
On the upper bound side, classic results give a $O(\log n)$-time algorithm for maximal matching~\cite{luby86, alon86}, which was later improved to $O(\log \Delta + \log \log n)$.
Through reductions to coloring, there is an $O(\Delta + \log^* n)$-time algorithm~\cite{panconesi-rizzi}.
For fractional matching, an $O(\log^2 \Delta)$ algorithm is known~\cite{Fischer17} leaving a polynomial gap to the lower bound.

\begin{table}
\scriptsize
\begin{tabular}{| c | c |c|}
  \hline 
 \textbf{Problem} & \textbf{Runtime} & \textbf{Paper} \\
 \hline 
 \hline 
 $(2\Delta - 1)$-edge col. &   $\Omega(\log^* n)$ & \cite{linial92} \\
 &  $\poly \log \Delta + O(\log^* n)$ & \cite{Balliu22-edge-coloring}  \\ 
  & $\widetilde{O}(\log^2 n)$  & \cite{GG23}\\
 &  $O(\log^2 \Delta \cdot \log n)$ & \cite{Ghaffari-rounding-2021} \\
 \hline
 \hline
 $(2\Delta - 2)$-edge col. & $\Omega(\log n)$ & \cite{CHLPU18}  \\
 \hline
 $(3\Delta/2 + \eps)$-edge col. & $O(\log^6 n)$  & \cite{Su-Vizing2019}+\cite{GGH23}+\cite{Splitting20}\\
 &  $\widetilde{O}(\eps^{-2}\cdot\log^3 n)$ & \textbf{this paper} \\
 &  $\widetilde{O}(\eps^{-2}\cdot\log^2 \Delta \cdot \log n)$ & \textbf{this paper}\\
 \hline
 $(3\Delta/2)$-edge col.  & $\widetilde{O}(\Delta^3\log^4n)$ & \cite{GKMU18} \\
 &  $O(\Delta^2 \cdot \log n)$ & \textbf{this paper}\\
 \hline
  $(\Delta + \eps\Delta)$-edge col. & $\widetilde{O}_{\eps}(\log^6 n)$ & \cite{Su-Vizing2019}+\cite{GGH23} \\
\hline
 $(\Delta + 1)$-edge col. &  $\widetilde{O}(\poly \Delta \log^5 n)$ & \cite{bernshteyn2024fastalgorithmsvizingstheorem,BERNSHTEYN2022319}  \\
 \hline
  \end{tabular}
\begin{tabular}{| c | c |c|}
  \hline 
 \textbf{Problem} & \textbf{Runtime} & \textbf{Paper} \\
 \hline 
 \hline 
 Saturating  &  $\Omega(\log n / \log \log n)$  & HSO LB\\
 Bipartite Matching &  $O(\log_{\delta / r} n)$ & \textbf{this paper}\\
 \hline
 \hline 
 HSO  &  $\Omega(n)$ for $\delta = r$ & \textbf{this paper}\\
 & $\Omega(\log_{\delta/r} n)$ & \cite{Brandt19,BBKO2021hideandseek}\\  
 & $O(\log_{\delta/r} n)$ & \textbf{this paper} \\  
 \hline 
 \hline
 Weak Splitting  &  $\Omega(\log n)$ & \cite{weakSplitting19}\\
 $\delta\geq 6r$ & $\widetilde{O}(\log^3 n)$ & \cite{weakSplitting19}\\
 $\delta \geq 2(r + 1)$ &   $O(\log_{\delta/r} n)$ & \textbf{this paper}\\
 \hline
\end{tabular}

\caption{The tables present a comparison of our results with prior results. \cite{Su-Vizing2019} provides a randomized $\widetilde{O}_{\eps}(\log^3 n)$-round algorithm for $(\Delta+\eps\Delta)$-edge coloring that can be derandomized via network decompositions \cite{RG20,GGH23}. The runtime of the $3\Delta/2$-edge coloring algorithm from \cite{GKMU18} is obtained by using the maximal independent set algorithm from \cite{GG23} as a subroutine for hypergraph matching. We use the notation  $\widetilde{O}$ to hide $\poly\log\log n$ factors. }. 
\label{table:EdgeColoring}
\end{table}

\subsection{Further Related Work}
\label{sec:furtherwork}

\paragraph{Sinkless Orientation.}
As mentioned before, sinkless orientation has played an important role in the development of various lines of research in the last decade.
In particular, sinkless orientation was the first local graph problem with a randomized intermediate complexity provably between $\omega(\log^*n)$ and $o(\log n)$\cite{LLL_lowerbound}, which initiated a stream of publications mapping out the complexity landscape of local graph problems in the \LOCAL model~\cite{naor95,balliu19lcl-decidability,LLL_lowerbound,BBOS18almostGlobal,BHKLOS18lclComplexity,binary_lcls,Brandt19,Brandt2017,Chang20,ChangKP19,Changhierarchy19,balliu21rooted-trees,Chang24}, but also in related models \cite{BCMOS21,MU21,BBFLMOU22}.
Similarly, understanding the complexity of sinkless orientation led to the development of the lower bound technique of round elimination~\cite{LLL_lowerbound,Brandt19}, which has dramatically advanced our understanding of distributed computation and led to a stream of lower bound results for many important problems (see, e.g., \cite{balliusupported24} for a comprehensive overview).

\paragraph{Weak Splitting.}
The sinkless and hypergraph sinkless orientation problems are closely connected to the weak splitting problem. In weak splitting, each node chooses either a blue or red color and the goal is to ensure that each node has at least 1 blue and 1 red node in its neighborhood.
The problem can be formulated through a bipartite graph where one side corresponds to the color choices and one side to the constraints.
It is known that the weak splitting is at least as hard as sinkless orientation as long as the input graph has minimum degree $5$~\cite{BGKMU19} and this works for any rank at least 2.
Furthermore, the splitting problem can be solved through a reduction to HSO if the rank is (at least) twice the degree.
Then, each node can split their hyperedges into two separate problem instances. In both instances, a node then gets at least one outgoing hyperedge.
To obtain the splitting each node can color its outgoing hyperedges using both colors at least once.

\paragraph{Randomized Covering and Packing.}
In recent work, Chang and Li gave randomized $O(\log n / \eps)$-time algorithms for solving integer linear programs (ILP), capturing problems such as maximum independent set, minimum dominating set, and minimum vertex cover~\cite{Chang-covering2023}.
While the results are not comparable to ours, the approaches have a common spirit.
As a building block, they design a ball-carving algorithm for the weak-diameter network decomposition, where a small fraction of nodes are allowed to be left unclustered with high probability.
Informally, the high probability guarantee allows to avoid repeating the process and hence, avoid getting higher polylogs.
In our deterministic case, this roughly corresponds to our technique of combining local solutions into a complete solution that is always correct.
Furthermore, there are almost matching lower bounds for the ILPs~\cite{goeoes14_DISTCOMP, podc16_BA, fractional-improved21}.

\subsection{Outline of the rest of the paper}
In \Cref{sec:preliminaries} we introduce the necessary notation, the concept of multihypergraphs required in our proofs as well as their bipartite representation. In \Cref{sec:HallDistributed}, we prove the Distributed Hall's Theorem (\Cref{thm:distributedHall}). Then, in \Cref{sec:HSO}, we use the theorem to design efficient algorithms for HSO, that is, we prove \Cref{thm:hypergraphSO} and \Cref{thm:randomizedHSO}.  In \Cref{sec:edgeColoring}, we present the edge coloring results (\Cref{thm:edgeColoring,cor:EdgeColoringFast}). The proof of the linear time lower bound for HSO appears in \Cref{sec:lowerbound}. 

\section{Preliminaries} \label{sec:preliminaries}
\paragraph{Multihypergraphs.} A \emph{hypergraph} is a generalization of a graph in which each edge can contain more than two vertices. A multihypergraph is a hypergraph in which the same edge may appear multiple times. If an edge appears more than once, all those copies of the same edge are considered as distinct edges. If not explicitly stated otherwise, all graphs in our work are multihypergraphs and sometimes we abbreviate the bulky term and simply speak of a graph. Similarly, often we just speak of an edge where we refer to a hyperedge. 

\paragraph{Notation.} 
For any (multihyper)graph $G$, let $V(G)$ denote its set of vertices, $E(G)$ the (multi)set of its edges. The $\mathsf{rank}(e)$ of an edge $e$ of a hypergraph is the number of vertices in the edge. The \emph{degree} $\deg(v)$ of a node in a hypergraph is the number of incident edges. We denote by $\delta(G)=\min_{v\in V(G)}\deg(v)$ a hypergraph's minimum degree, and by $\mathsf{rank}(G)=\max_{e\in E(G)}\mathsf{rank}(e)$ its maximum rank. The neighborhood $N^G(v)$ of a node in a hypergraph consists of the vertices that share a hyperedge $v$. For some vertex set $S\subseteq V(G)$ the node-induced subgraph $G[S]=(S,\{e\in E(G)\mid e\subseteq S\}$ is the graph on vertex set $S$ that contains all edges with all endpoints in $S$.
For some integer $x\geq 0$, multihypergraph $G$ and vertex $v\in V(G)$, we use $B^G_x(v)$ to denote the subgraph induced by all vertices in hop-distance at most $x$ in $G$ from $v$. 

\paragraph{Bipartite representation of a multihypergraph.} We sometimes consider the hypergraphs as hypergraphs and sometimes as bipartite graphs with a vertex side and a hyperedge side (where the hyperedge-side nodes correspond to hyperedges of the hypergraph). 
 The bipartite representation $\cB_G=(V,E,F)$ of a hypergraph $G=(V,E)$ consists of vertex sets $V$ and $E$ and has an edge in $F$ between $v\in V$ and $e\in E$ if and only if $v\in e$. As the edge set $F$ is always implicitly given, we often omit it in our notation. For the sake of presentation and to guide the reader through our proofs, we refer to the two partitions of a bipartite graph as the \emph{vertex side} and as the \emph{hyperedge side}. The degree of a node on the vertex side corresponds to the degree of the node in the multihypergraph, and the degree of a hyperedge in the bipartite graph corresponds to the number of vertices in that hyperedge, that is, to its rank. Note that the bipartite representation of a multihypergraph is a simple graph in which each edge has multiplicity one. Multiple parallel edges in a hypergraph $G$ appear as multiple hyperedge nodes on the hyperedge side in $B_G$.
The neighborhood of a vertex $v$ in the bipartite representation $\cB_G$ of a hypergraph $G$ consists of the edges that are incident to $v$. 

\paragraph{Matchings in bipartite graphs.} A matching $M\subseteq F$ in a bipartite graph $\cB=(V,E,F)$ is an independent set of edges. A \emph{node-saturating matching} in a bipartite graph $\cB=(V,E,F)$ is a matching $M\subseteq F$ such that every node on the vertex side is matched. 

\begin{restatable}{observation}{HSOisSaturating}
    \label{obs:bipartiteHypergraph}
Let $G=(V,E)$ be a hypergraph. Any node-saturating matching in $\cB_G=(V,E,F)$ corresponds to an HSO of $G$ and vice versa.
\end{restatable}
\begin{proof}
    Given a node-saturating matching $M$ in $\cB_G=(V,E,F)$,  each hyperedge $e$ is oriented outwards from every node $v$ if $(e,v)\in M$. As $M$ is node-saturating, we obtain that every node $v$ of the hypergraph $G$ has an outgoing hyperedge. All other hyperedges are oriented arbitrarily satisfying the constraint that each hyperedge is outgoing for at most one of its nodes. 

    By definition, HSO ensures that each node $v$ is contained in at least one hyperedge, where $v$ is incoming and all other nodes are outgoing.
    Hence, a node-saturating matching is obtained by each node $v$ selecting exactly one hyperedge, where $v$ is incoming and all other nodes are outgoing.
\end{proof}

We will use the following ``one-sided'' version of Hall's theorem.
\begin{theorem}[Hall's Theorem~\cite{Halls}]
\label{thm:Hall}
A bipartite graph $\cB(V,E,F)$ admits a node-saturating matching if and only if $|N(S)|\geq |S|$ for all $S\subseteq V$.
\end{theorem}

Recall that a hypergraph has \emph{rank} $r$ if each hyperedge contains at most $r$ nodes. 

\begin{lemma}
\label{lem:ExistentialEdgeGrabbing}
    Any  hypergraph with minimum degree $\delta$ and maximum rank $r\leq \delta$ admits an HSO. 
\end{lemma}
\begin{proof}
For a hypergraph $G=(V,E)$ let $\cB_G(V,E,F)$ be its bipartite representation. Consider an arbitrary set of vertices $S\subseteq V$. As every node in $S$ has $\delta$ incident edges, and each edge contains at most $r$ vertices we obtain $|N(S)|\geq \delta/r \cdot |S|$. By Hall's Theorem (\Cref{thm:Hall}), $\cB_G=(V,E,F)$ has a node-saturating matching which implies an HSO on $G$ due to \Cref{obs:bipartiteHypergraph}.
\end{proof}

\section{Distributed Hall's Theorem}
\label{sec:HallDistributed}
In this section, we prove our distributed version of Hall's Theorem, i.e., \Cref{thm:distributedHall}.
The proof of the theorem follows the roadmap as explained in \Cref{sec:nutshell}.

First, we begin with one of our central definitions, that is, a Hall graph. 
\begin{definition}[Hall graph] A multihypergraph $H$ is a \emph{Hall graph} if $\mathcal{B}_H$ admits a node-saturating matching, or in other words, if $H$ has an HSO.
\end{definition}

The next lemma is the heart of our approach; it shows that every vertex is contained in a small-diameter Hall graph.

\begin{restatable}{lemma}{lemManyEdges}
\label{lemma:manyEdges}
 For $\delta\geq r>0$ let $x(n)= \log_{\frac{\delta-1}{r-1}}n$ and let $G=(V,E)$ be a multihypergraph with minimum degree $\delta$ and maximum rank $r$.
Then for any $v\in V$ there exists a subgraph $G'(v)\subseteq B^G_{x(n)}(v)$  with $v\in V(G'(v))$ and  $|E(G'(v))|\geq |V(G'(v))|$. 
\end{restatable}
\begin{proof} We begin with a standard ball growing argument that only takes the number of vertices into account. 
\begin{claim}
\label{claim:findVertexBall}
There exists $0\leq x'\leq x(n)$ such that $|V(B^G_{x'+1}(v))|\leq \frac{\delta-1}{r-1}|V(B^G_{x'}(v))|$ holds. 
\end{claim}
\begin{proof}
Assume for contradiction that no $0\leq x'\leq x(n)$ satisfies the condition. Let $\alpha=\frac{\delta-r}{r-1}$ and observe that $\frac{\delta-1}{r-1}=1+\alpha$. Then, we have
\begin{align*}
|V(B_{x(n)}^G(v))|> (1+\alpha)|V(B_{x(n)-1}^G(v))|\geq (1+\alpha)^{x(n)}|V(B_0^G(v))|= \left(\frac{\delta-1}{r-1}\right)^{x(n)}=n~.
\end{align*}
As $B_{x(n)}^G(v)\subseteq G$ can contain at most $n$ vertices, this is a contradiction. 
\renewcommand{\qed}{\ensuremath{\qquad\blacksquare}}
\end{proof}

First apply \Cref{claim:findVertexBall} to find some $0\leq x\leq x(n)$ satisfying the conditions in the claim. 
Let $N=|V(B^G_{x}(v))|$ denote the number of nodes of  $B^G_{x}(v)$, and $N'=|V(B^G_{x+1}(v))|-|V(B^G_{x}(v))|$ the number of nodes of  $B^G_{x+1}(v)$ that are not nodes of $B^G_{x}(v)$.
By the properties of the \Cref{claim:findVertexBall}, we have
\begin{align}
    \label{eqn:NNprime}
   N + N' \leq \frac{\delta-1}{r-1} \cdot N. 
\end{align}

Now let $G'=G'(v)$ be the  subgraph that contains all edges with at least one endpoint in $V(B_x^G(v))$. Note that the vertex set of $G'$ is a subset of $V(B_{x+1}^G(v))$. 

For any node $u \in V(B^G_{x}(v))$, the number of hyperedges in $G'$ incident to $u$ is at least $\delta$.
For any node $u \in V(B^G_{x+1}(v)) \setminus V(B^G_{x}(v))$, the number of hyperedges in $G'$ incident to $u$ is at least $1$.
Hence, the sum of the ranks of the hyperedges in $G'$ is at least $\delta \cdot N + N'$. Thus, $G'$ has at least $\frac{N\delta+N'}{r}$ hyperedges. All vertices of $G'$ live in $B^G_{x+1}(v)$, and hence $V(G')\leq N+N'$. 
    
We have that the number of nodes $V(G')$ is smaller or equal to the number of edges $E(G')$ if the following series of equivalent statements holds
\begin{align*}
V(G')\leq N+N'&\leq \frac{N\delta+N'}{r}\leq E(G') &  \mid \text{multiply by $r$}\\
    r(N+N')&\leq N\delta+N'  &  \mid \text{subtract $N+N'$}\\
    (r-1)(N+N')&\leq (\delta-1) N &  \mid \text{divide by $r-1>0$}\\
    N+N'&\leq \frac{\delta-1}{r-1} \cdot N &
\end{align*} 
The last statement holds due to Inequality~\ref{eqn:NNprime}, implying that $V(G')\leq E(G')$.
\end{proof}

Next, we turn our attention to the useful notion of a Hall violator, and then proceed with one of the lemmas outlined in the roadmap for the proof of \Cref{thm:distributedHall}.

\begin{definition}[Hall violator]
A \emph{Hall violator} of a multihypergraph $G=(V,E)$ is a set of nodes $S\subseteq V$ such that $|N_{\cB_G}(S)|<|S|$ holds, or in other words, such that the number of edges with at least one of its vertices in $S$ is strictly smaller than $S$.
\end{definition}

\begin{claim}
\label{claim:HallViolatorExists}
Any multihypergraph that is not a Hall graph contains a Hall violator. 
\end{claim}
\begin{proof}
By Hall's Theorem (\Cref{thm:Hall}), if a multihypergraph $G$ is not a Hall graph, then there is some set of nodes $S\subseteq V$ such that $|N_{\cB_G}(S)|<|S|$.
\end{proof}

\begin{restatable}{lemma}{lemNonEmptyHallExists}
\label{lem:nonEmptyHallExists}
Any non-empty multihypergraph $G=(V,E)$ with $|V|\leq |E|$ contains a non-empty Hall subgraph.
\end{restatable}
\begin{proof}
In this existential proof, we iteratively produce a decreasing sequence $G=R_0\supseteq R_1\supseteq\ldots \supseteq R_k\neq\emptyset$ of subgraphs of $G$ such that $R_k$ is the desired non-empty Hall subgraph. We start with $R_0=G$. Given, some $R_i$, the procedure stops if $R_i$ is a Hall graph. Otherwise, we construct $R_{i+1}$ from $R_i$ as follows. By \Cref{claim:HallViolatorExists} there exists some node set $S_i$ that is a Hall violator of $R_i$. In that case let $G_{i+1}=(V(G_i)\setminus S_i, E(R_i)\setminus \{e\in E(R_i)\mid e\cap S_i\neq\emptyset\})$ be the subgraph of $R_i$ that we obtain by removing all vertices in $S_i$ and all edges with at least one endpoint in $S_i$ from the multihypergraph $R_i$. As $S_i$ is a Hall violator, this process removes strictly more vertices than edges and by induction hypothesis ($(V(R_i)\leq E(R_i)$) we obtain $V(R_{i+1})<E(R_{i+1})$ for all $i\geq 0$. Hence, the process cannot end with an empty graph and returns a non-empty Hall graph that is a subgraph of $G$. 
\end{proof}

Before finally starting with the proof of \Cref{thm:distributedHall}, we need a last technical definition.

\begin{definition}
    For a set $V$ of vertices and a multiset $E = \{ e_1, \dots, e_k \}$ of edges, let $E|_V$ denote the multiset $\{ e'_1, \dots, e'_k\}$ defined by $e'_i := e_i \cap V$ for any $1 \leq i \leq k$.
\end{definition}

Now we are set to prove the main technical contribution of our work.

\medskip
\noindent
\textit{Proof of \Cref{thm:distributedHall}.}\quad
Let $G = (V, E)$ be an $n$-node graph, and consider a vertex $v \in V$.
Set $x := x(n) :=  \log_{\frac{\delta - 1}{r - 1}} n \leq \log_{\frac{\delta}{r}} n$.
We prove the theorem by providing a (sequential) algorithm that computes a Hall graph with the desired properties.
    
We start by defining a sequence $G = G_0, G_1, G_2, \dots, G_k$ of multihypergraphs on decreasing sets of vertices, all of which will contain $v$ and satisfy that the minimum degree $\delta(G_i)$ is strictly larger than the maximum rank $r(G_i)$ and a sequence $H_0,\ldots,H_k$ of Hall (multihyper)graphs such that $H_i\subseteq G_i$. Informally, our algorithm returns $H=H_0\cup \ldots\cup H_k$  and we show that $H$ is a small-diameter Hall graph that contains $v$; the formal construction of the returned graph $H$ appears at the very end of this proof.

\begin{algorithm}
\caption{Distributed Hall's Theorem (informal notation)}\label{alg: violators}
\begin{algorithmic}
\setlength\itemsep{1.5pt}
\State Set $i \coloneqq 0$.
\State Set $G_0(v) \coloneqq G\left[ B_{x(n)}^{G} (v) \right]$.
\Do ~(increase $i$ in each iteration)
 \State \Cref{lemma:manyEdges}: Find a subgraph $G_i'(v)\subseteq G_i(v)$ with $v\in V(G_i')$,
  and  $|E(G_i'(v))|\geq| V(G_i'(v))|$.
 \State \Cref{lem:nonEmptyHallExists}:  Find Hall subgraph $H_i\subseteq G_i'(v)$
 \State Set $G_{i + 1}(v) \coloneqq G_i(v) \setminus H_i$.
\doWhile{$v \not \in H_i$}\\
\Return $H_0\cup\ldots \cup H_{i}$
\end{algorithmic}
\label{alg:DistributedHall}
\end{algorithm}

\paragraph{Constructing $G_{i+1}$ from $G_i$. }For any $0 \leq i \leq k - 1$, we obtain $G_{i + 1}$ from $G_i$ as follows. First, use  \Cref{lemma:manyEdges} to find subgraph $G'_i$ of $G_i$ that has diameter at most $x$, contains $v$, and satisfies $E(G'_i) \geq V(G'_i)$.   Then, use \Cref{lem:nonEmptyHallExists} to find a nonempty subgraph $H_i$ of $G'_i$ (and therefore also of $G_i$) that is a Hall graph. We can apply \Cref{lem:nonEmptyHallExists} because  $G'_i$ is nonempty as it contains $v$.
If $v \in V(H_i)$, set $k := i$ (i.e., the construction of the sequence of graphs is concluded) and abort computing $G_{i + 1}$.
Otherwise, define $G_{i + 1}$ by setting $V(G_{i + 1}) := V(G_i)\setminus V(H_i)$ and $E(G_{i + 1}) := \left(E(G_i) \setminus E(H_i)\right)|_{V(G_{i + 1})}$.
In other words, we obtain $G_{i + 1}$ from $G_i$ by first removing all vertices and edges that also occur in $H_i$ and then removing \emph{from each remaining edge} all vertices that have been removed from the graph. 

\paragraph{Why can we apply \Cref{lemma:manyEdges} on \boldmath$G_i$?} We note that the fact that $G_{i + 1}$ is computed only if $v \notin V_{H_i}$ holds implies that $V(G_{i+1}) = V(G_i)\setminus V(H_i)$ contains $v$ (as $v \in V(G_i)$ by inductive assumption).    
 The degree of any node $v \in V(G_{i + 1})$ is the same in $G_{i + 1}$ and $G_i$, because for any edge of $G_i$ that is not present in $G_{i+1}$ we also remove all of its vertices when removing the subgraph $H_i$. The rank of an edge $e \in E(G_{i+1})$ is at most as large the rank of the corresponding edge of $G_i$. Thus, we obtain $\delta(G_{i + 1}) \geq \delta(G_i)$ and $r(G_{i + 1}) \leq r(G_i)$, which implies $\delta(G_{i + 1}) > r(G_{i + 1})$, due to $\delta(G_i) > r(G_i)$.
This shows inductively that the $G_i$s are well-defined multihypergraphs whose minimum degree exceeds the maximum rank. 
   
Furthermore, since the construction of the graph sequence ensures that $|V(G_{i+1})| < |V(G_i)|$ (as $H_i$ is nonempty), we know that the construction terminates at some point (due to every $V(G_i)$ being contained in $V(B_{x}^G(v))$), implying that the sequence is indeed finite and $k$ well-defined.
    
Note that the definitions of the $G_i$ depend on our choices in the described sequential algorithm.
If desired, those choices can be fixed by taking the lexicographically first object out of a choice of objects whenever there is a choice.

\paragraph{Keeping track of edges.} As we later want to ``lift' the graph $H_i$ back to the original graph $G$, we formally keep track of the modifications of the edges (consisting of removing endpoints) via a projection (aka function) $\pi_i: E(G_i)\rightarrow E(G_{i+1})$: for each edge $e \in E(G_i)$, $\pi_i(e)$ denotes the projection of $e$ to the vertices of $V(G_{i+1})$.
In particular, $\pi_i$ is a bijection between the two multisets $E(G_i)$ and $E(G_{i + 1})$.
For an edge $e \in E(G_{i + 1})$, we will call the edge $\pi_0^{-1} \circ \pi_1^{-1} \circ \dots \circ \pi_{i}^{-1}(e) \in E(G)$ the \emph{original} edge corresponding to $e$, where $\pi_0^{-1} \circ \pi_1^{-1} \circ \dots \circ \pi_{i-1}^{-1}(\cdot)$ is the function defined by iteratively applying $\pi_{i-1}^{-1}, \dots, \pi_0^{-1}$ (in that order).
In other words, the \emph{original} edge corresponding to $e\in E(G_i)$ is the edge from $E(G)$ that resulted in $e$ by removing vertices (via applying the functions $\pi_j$) during the iterative construction of our graph sequence.

\paragraph{Computing the final Hall graph $H$:}    Now we explain how we derive the desired Hall $H$ graph from the computed Hall graphs $H_0, \ldots, H_k$. To this end, for each $0 \leq i \leq k$, define $E^*(H_i) := \{ \pi_0^{-1} \circ \pi_1^{-1} \circ \dots \circ \pi_{i-1}^{-1}(e) \mid e \in E(H_i) \}$.
    In other words, $E^*(H_i)$ is the set of original edges corresponding to the edges contained in $E(H_i)$.
    In particular, $E^*(H_i) \subseteq E(G)$. We define the  final Hall graph   as
    \begin{align*}
        H := \left( \bigcup_{0 \leq i \leq k} V(H_i), \bigcup_{0 \leq i \leq k} E^*(H_i) \right) \ .
    \end{align*} 
    We continue with proving that $H$ is the desired Hall graph. 
   \textbf{$H$ is a well-defined multihypergraph:} The construction of the $G_i$s ensures that for each $0 \leq j \leq k$, and each $e \in E(H_j)$, the original edge $e^*\in E(G)$ corresponding to $e$ contains only endpoints in $\bigcup_{0 \leq i \leq j} V(H_i)$. Thus, for  any edge $e^* \in \bigcup_{0 \leq i \leq k} E^*(H_i)$, all endpoints of $e^*$ are contained in $\bigcup_{0 \leq i \leq k} V(H_i)$, showing that $H$ is a well-defined hypergraph.     \textbf{\boldmath$v\in V(H)$:}
    From the construction of the $G_i$, it follows that $v \in V(H_k)$, i.e., $v$ is contained in the vertex set of the $H_i$ computed last.
    Hence, $v \in V(H)$ as desired.
    Our next objective is to show that $H$ is indeed a Hall graph.
\textbf{$H$ is a subgraph of $B_x^G(v)$:} We have already reasoned that $V(H)\subseteq V(G)$ and $E(H)\subseteq E(G)$. Hence, the claim follows from the facts that, for any $0 \leq i \leq k$, $H_i$ is a subgraph of $G'_i$, and $G'_i$ has diameter at most $x$ and contains $v$.

\textbf{$H$ is a Hall graph:}
    We start by observing that, by the construction of the $V(G_i)$, we have $V(H_i) \cap V(H_j) = \emptyset$ for any $0 \leq i < j \leq k$.
    Moreover, by the construction of the $E(G_i)$, we know that for any $0 \leq i < j \leq k$ and any two edges $e \in E(H_i)$ and $e' \in E(H_j)$, the ``original'' two edges $\pi_0^{-1} \circ \pi_1^{-1} \circ \dots \circ \pi_{i-1}^{-1}(e)$ and $\pi_0^{-1} \circ \pi_1^{-1} \circ \dots \circ \pi_{j-1}^{-1}(e')$ corresponding to $e$ and $e'$, respectively, in $E(G)$ are distinct.
    An analogous statement holds for the case that $e$ and $e'$ are distinct edges from the same $E(H_i)$.
    (In all of these statements, whenever we say ``distinct'', the two edges can be parallel edges (i.e., have the exact same set of endpoints) but cannot refer to the same element in the multiset.)
    Furthermore, by construction, each $H_i$ is a Hall graph, which implies for each $0 \leq i \leq k$ that the bipartite graph $\mathcal{B}_G(V(H_i),E(H_i))$ admits a node-saturating matching.
    For each $0 \leq i \leq k$, fix such a node-saturating matching on $\mathcal{B}_G(V(H_i),E(H_i))$, and for each node $v \in V(H_i)$, let $M_i(v)$ denote the edge $e \in E(H_i)$ that $v$ is matched to.
    Now we define a node-saturating matching on $\mathcal{B}_G(\bigcup_{0 \leq i \leq k} V(H_i), \bigcup_{0 \leq i \leq k} E^*(H_i))$ by matching any $v \in \bigcup_{0 \leq i \leq k} V(H_i)$ to an edge $e \in \bigcup_{0 \leq i \leq k} E^*(H_i)$ as follows.
    Let $i$ be the (uniquely defined) index such that $v \in V(H_i)$.
    Set $e := \pi_0^{-1} \circ \pi_1^{-1} \circ \dots \circ \pi_{i-1}^{-1}(M_i(v))$, i.e., we match $v$ to the original edge of $G$ corresponding to $v$'s matching partner in the matching on $H_i$.

    By the above considerations, it follows that the defined matching is a matching, and as it is also node-saturating on $\mathcal{B}_G(\bigcup_{0 \leq i \leq k} V(H_i), \bigcup_{0 \leq i \leq k} E^*(H_i))$, we obtain that $H$ is indeed a Hall graph. \qed

\section{Hypergraph Sinkless Orientation}

\subsection{Deterministic Algorithm for HSO}
\label{sec:HSO}

\thmHSO*
\begin{proof}
We first describe a sequential process based on Distributed Hall's Theorem (\Cref{thm:distributedHall}) to solve the problem. 
Let $x=\log_{\frac{\delta-1}{r-1}}n$. 
Apply \Cref{thm:distributedHall} for each node $v\in V$ to find a Hall graph $H(v)$ that is a subgraph of $B_x^G(v)$ and contains $v$. 
Order these Hall graphs according to the IDs of the vertices $H_1=H(v_1), H_2=H(v_2), \ldots, H_n=H(v_n)$. Now, sequentially iterate through the Hall graphs, and when processing $H_i$ orient all hyperedges in $E(H_i)$ according an arbitrary HSO orientation of $H_i$ (which exists as $H_i$ is a Hall graph). Note that this process re-orients each hyperedge that was already oriented while processing $H_1, \ldots, H_{i-1}$. At the very end orient all hyperedges that do not appear in any of the Hall graphs arbitrarily. We prove that the computed hyperedge orientation is an HSO. Consider an arbitrary vertex $v$ and let $H_{i_v}$ be the Hall graph with largest index that contains $v$; note that such a Hall graph has to exist as $H(v)$ contains node $v$. Hence, there is a hyperedge $e\in E(H_{i_v})$ that was oriented outwards from $v$ when processing Hall graph $H_{i_v}$. As $v$ does not appear in any  Hall graph with a larger index also edge $e$ does not appear in any of these (all considered Hall graphs are subgraphs of $G$ and an edge can only be contained in it if all of its vertices are). Thus, edge $e$ does not change its orientation after processing $H_{i_v}$ and it is oriented outwards from $v$ in the final orientation. 

In the \LOCAL model, we can simulate this sequential algorithm as each computed Hall graph is contained in the radius-$x$ ball around the respective ``center'' node. For each of the edges of the hypergraph assign one of its endpoints as the responsible node to orient the hyperedge. A node orients the edges that it is responsible for as follows: First, each node queries its $2x$-hop neighborhood and computes $H(u)$ for all nodes in $B_x^G(v)$. Observe that, by the properties of \Cref{thm:distributedHall}, $H(u)$ is a subgraph of $B_x^G(u)\subseteq B_{2x}^G(v)$.
No hall graph $H(u)$ for some $u\notin V(B_x^G(v))$ can contain an edge incident to $v$. Knowing, the identifiers of nodes of $B_x^G(v)$, node $v$ can for each  incident edge $e$ compute the Hall graph $H_{i_e}$ with the largest index containing $e$ and orient the according to the HSO of $H_{i_e}$; of there is no such index, edge $e$ is oriented arbitrarily. This algorithm requires that all nodes that process some Hall graph $H(v)$ use the same HSO orientation for $H(v)$. This is not a strong requirement as one can just use the lexicographically smallest HSO orientation according to an arbitrary order of all feasible HSO orientations of $H(v)$. This process orients each edge as in the sequential process and we obtain an HSO of $G$. 
The runtime of the process is $2x$ as there is no communication after learning the $2x$-hop balls.
\end{proof}

\corXYSplitting*
\begin{proof}
    First, we modify the input instance as follows.
    For each node $v$ on the left side, we create two virtual copies $v_1$ and $v_2$.
    Half of the neighbors of $v$ are connected to $v_1$ and the rest to $v_2$, arbitrarily.
    Notice that the minimum degree $\delta'$ of the new instance is at least $r + 1$. 

    Using \Cref{thm:hypergraphSO}, we find an HSO in the modified graph in $O(\log_{\frac{\delta' - 1}{r - 1}} n) = O(\log_{\frac{\delta/2 - 1}{r - 1}} n) = O(\log_{\frac{\delta}{r}} n)$ time, which assigns at least one hyperedge to both $v_1$ and $v_2$.
    Now, $v_1$ can color its outgoing edges with one color and $v_2$ with the other.
    Hence, in the original graph, node $v$ has at least one of each color in its neighborhood. 
 \end{proof}

\subsection{Randomized Algorithms for HSO}
\label{sec:randomized}

The HSO problem can be modeled as an an instance of the Constructive \lovasz Local Lemma (LLL) \cite{CPS17,FGLLL17,RG20}:  Orient each hyperedge uniformly at random, i.e., the hyperedge is outgoing for a single of its endpoints selected uniformly at random. Each hyperedge makes a node \emph{happy} independently with probability at least $1/r$. 
By a simple reduction, we can consider the case where each node has a degree of exactly $\delta$.
Then the probability for a node to be unhappy is  $p=(1-1/r)^{\delta}\leq e^{-\Omega(\delta/r)}$. 
Furthermore, each such ``bad event'' depends on at most $\delta \cdot r = \poly(\delta)$ other bad events, and hence the problem is an LLL with dependency degree $\delta \cdot r$ and bad event probability $p$, which satisfies a polynomial LLL criterion if $\delta\geq c r\log r$ for a sufficiently large constant $c>1$. 
With that, one can use existing LLL algorithms to compute random choices for each edge that avoid all bad events, i.e., give an outgoing hyperedge for every node. 
There are randomized algorithms that run either in $O(\poly(\delta,r)+\poly\log\log n)$~\cite{FGLLL17,RG20} or in $O(\log n)$ rounds~\cite{CPS17}. 

\thmHSORandomized*
\begin{proof}
First,  we assume that $r\leq 10\log n$, and we also let each vertex drop out of all but $320r\log r$ hyperedges, reducing the maximum degree of the graph to $\poly\log n$. 
In fact, for the rest of the proof we assume that the graph is $\delta$-regular with $\delta=320r\log r$.
Note that we also obtain that all node degrees are the same, but some hyperedges may have fewer than $r$ vertices. The rest of the proof uses the shattering framework, similarly to the algorithm in \cite{GS17}, but with our new deterministic algorithm from \Cref{thm:hypergraphSO} in the post-shattering phase. 

\paragraph{Pre-shattering.} Each edge activates itself with probability $1/4$, activated edges point outwards from one of their up to $r$ chosen nodes uniformly at random. 
\begin{itemize}
\item Let $Bad_1$ be the nodes who have more than $\delta/2$ or fewer than $\delta/8$ incident activated hyperedges. 
\item Let $Bad_2$ be the nodes that are not in $Bad_1$ but have a neighbor in $Bad_1$. 
\item Let $Bad_3$ be the nodes $\notin Bad_1\cup Bad_2$ that do not have an outwards oriented hyperedge. 
\end{itemize} 
Nodes in $Bad_1$ deactivate all their edges and undo their orientation. 
Let $B=Bad_1\cup Bad_2\cup Bad_3$

\begin{lemma}
  \label{lem:shatteringProbability}
    For any node $v\in V$ probability that $v\in B$ is at most $1/\delta^{20}$.  For each hyperedge $e\in E$, the probability that one of its vertices is contained in $B$ is upper bounded by $1/\delta^{20}$. All these events are independent for nodes and hyperedges that are at least $4$ hops apart in the bipartite representation. 
\end{lemma}
\begin{proof}
    If $r,\delta$ are larger than a sufficiently large constant, we obtain for $v\in V$ that $Pr(v\in Bad_1)\leq \exp(-\delta/12)\leq \delta^{-22}$ via a Chernoff bound and $Pr(v\in Bad_2)\leq \sum_{u\in N(v)} Pr(u\in Bad_1)\leq \delta\cdot r\cdot \delta^{-22}\leq \delta^{-20}$. We also obtain $Pr(v\in Bad_3)\leq (1-1/r)^{\delta/8}\leq e^{-40\log r}\leq \delta^{-20}$. The last inequality follows as $\delta=320r\log r\leq r^2$.
\renewcommand{\qed}{\ensuremath{\qquad\blacksquare}}
\end{proof}

\begin{lemma}[The Shattering Lemma, \cite{FGLLL17,BEPSv3}]\label{lem:Shattering}
Let $G=(V, E)$ be a graph with maximum degree $\Delta$. Consider a process which generates a random subset $B \subseteq V$ such that $P[v \in B]\leq \Delta^{-c_1}$, for some constant $c_1 \geq 1$, and such that the random variables $1(v\in B)$ depend only on the randomness of nodes within at most $c_2$ hops from $v$, for all $v\in V$, for some constant $c_2\geq 1$.
Then, for any constant $c_3\geq 1$, satisfying  $c_1>c_3+ 4c_2 + 2$,  we have that any connected component in $G[B]$ has size at most $O( \log_{\Delta} n  \Delta^{2c_2})$ with probability at least $1- n^{-c_3}$.
\end{lemma}

\paragraph{Post-shattering.} The post-shattering instance consists of all vertices in $B=Bad_1\cup Bad_2\cup Bad_3$ and all hyperedges with at least one endpoint in $B$, but restricted to the nodes in $B$. Let $\cB_G^{\mathsf{post}}=(B,E',F)$ be the resulting bipartite representation of the graph. 

\begin{lemma}
W.h.p. each connected components of $\cB_G^{\mathsf{post}}$ has $O(\poly r \log n)$ nodes, each node on the vertex side has degree at least $\delta/2$ and each node on the hyperedge side has degree (rank) at most $r$.  
\end{lemma}
\begin{proof}
The bound on the rank of the hyperedges follows as each hyperedge is a subset of a hyperedge of the original bipartite graph. The bound on the minimum degree on the vertex side follows as every vertex in $Bad_1\cup Bad_2\cup Bad_3$ has at least $\delta/2$ unmarked incident hyperedges after the pre-shattering phase; note that the nodes that had fewer unmarked incident hyperedges actually are in $Bad_1$ in the first step of the first phase and then unmarked all their incident hyperedges. 

Recall, that we are in the setting where each node has degree $\delta$ and we have $\delta \geq r$, that is, the bipartite representation of $G$ has maximum degree $\delta$.
The claim on the connected component size of $O(\log_{\delta} n \cdot \delta^{8})\leq O(\poly r \log n)$ follows with high probability via the Shattering Lemma (\Cref{lem:Shattering})  and \Cref{lem:shatteringProbability} applied on the bipartite representation $\cB_G$ of $G$ with $c_1=20$, $c_3=1$, $c_2=4$. 
\renewcommand{\qed}{\ensuremath{\qquad\blacksquare}}
\end{proof}

The final runtime follows by applying \Cref{thm:hypergraphSO} on all connected components of $\cB_G^{\mathsf{post}}$ in parallel. As each of them has at most $N=O(\poly r \log n)$ nodes, nodes have minimum degree $\delta/2$ and hyperedges have maximum rank $r$, the runtime is $O(\log_{\frac{\delta/2-1}{r-1}}N)=O(\log_{\frac{\delta}{r}}\delta+\log_{\frac{\delta}{r}}\log n)$. 
Each node participating in the post-shattering phase receives at least one outgoing edge from the application of \Cref{thm:hypergraphSO}, and each node not participating in the post-shattering phase has at least one outgoing edge from the pre-shattering phase. 

\paragraph{Preprocessing for alternative algorithm $(r\geq 100\log n)$.} First of all, we reduce the number of hyperedges to at most $n^3$, as we let  every node vote for $n^2$ of its incident hyperedges and we remove any hyperedge without a vote. 
If $r\geq 100\log n$, each node remains in each of its incident hyperedges independently (for each incident hyperedge) with probability $p=(100\log n)/r$. Let $G'$ be the resulting graph. In expectation the rank for any hyperedge $e$ is at most $p\cdot r\leq 100\log n$, and with a Chernoff bound, the probability that it is above  $200\log n$ is at most $1/\poly n$. As there are at most $n^3$ hyperedges, w.h.p.\ (in $n$) all hyperedges have rank at most $200\log n$. Similarly, the expected degree of a node in $G'$ is $p\cdot \delta$ and with a Chernoff bound the probability that it is below $p\cdot \delta/2$ is at most $1/\poly n$. With a union bound over the $n$ vertices w.h.p.\ (in $n$), all vertices have degree at least $\delta'=p\cdot \delta/2\geq 100\log n\cdot \delta/(2r)\geq 4000\log n \log r\geq 320 r'\log r'$. Nodes with a larger degree, simply leave the appropriate number of hyperedges such that we obtain a $\delta'$-regular hypergraph with maximum rank $r'$ satisfying $\delta'\geq 320 r'\log r'$. Now, we run the previous algorithm which takes $O(\log_{\frac{\delta'}{r'}}\delta'+\log_{\frac{\delta'}{r'}}\log n)=O(\log\log n/\log\log\log n)$.
\end{proof}

\section{Edge Coloring}
\label{sec:edgeColoring}
The main objective of this section is to prove the following theorem.  
\thmedgeColoring*

In \Cref{sec:edgeColoringProof}, we also prove \Cref{cor:EdgeColoringFast} that provides a faster algorithm for coloring with $(3/2+\eps)\Delta$ colors. 

\Cref{thm:edgeColoring}'s edge coloring algorithm is based on the  edge coloring framework of \cite{GKMU18}.  In the sequel, we sketch this framework; see Algorithm~\ref{alg:edge coloring} for pseudocode. Their algorithm computes a $3\Delta/2$-edge coloring in  $O(\Delta^8\log^9n\log^5 \Delta)$ rounds while we aim for an $O(\Delta^2\log n)$-round algorithm. 

The crucial definition of their approach is a so-called (3)-graph. 
\begin{definition}[(3)-graphs]
\label{def:threeGraph}
    A \emph{(3)-graph} is a graph with maximum degree 3 where no two degree-3 vertices are adjacent.
\end{definition}

To obtain their edge coloring algorithm, they iteratively extract (and remove) (3)-graphs from $G$ in a way that reduces the maximum degree of the remaining graph by at least two in each iteration.  Each of these (3)-graphs can be edge colored (in parallel) with 3 colors and using a fresh set of colors for each of the $\approx \Delta/2$ extracted graphs yields a $3\Delta/2$-edge coloring of $G$.  The runtime of their procedure depends on two factors: how fast one can extract a single $(3)$-graph and how fast one can color (3)-graphs. In order to obtain a logarithmic-time algorithm (on constant-degree graphs), we improve the runtime for both ingredients. Our left side saturating matching algorithm under the condition $\delta>r$ will be the crucial ingredient for improving the extraction process (see \Cref{sec:threeGraphExtraction}) while we develop an entirely new algorithm for edge coloring the (3)-graphs (see \Cref{sec:colorThreeGraphs}). Both of these results are summarized in the following lemmas. 

\begin{lemma}
\label{lem:threeGraphExtraction}
 There is a deterministic \LOCAL algorithm with time complexity
    $O(\Delta \cdot \log n)$
     that for any $n$-node graph $G=(V,E)$ with maximum degree $\Delta\geq 3$ computes an edge set $F \subseteq E$ such that $H = (V,F)$ is a (3)-graph and the maximum degree of the graph $(V,E - F)$ is at most $\Delta - 2$.
\end{lemma}

\begin{restatable}{lemma}{lemColorThreeGraphs} \label{lem:ColorThreeGraphs}
There is a deterministic \LOCAL algorithm that computes a $3$-edge coloring of any $n$-node (3)-graph in $O(\log n)$ rounds. 
\end{restatable}
These lemmas are proven in \Cref{sec:threeGraphExtraction,sec:colorThreeGraphs}, respectively. 

\medskip
\noindent
\textit{Proof of \Cref{thm:edgeColoring}.}\quad
Let $G = (V,E)$ be an undirected graph with maximum degree $\Delta$. We then apply $k = \lfloor \frac{\Delta -1}{2} \rfloor$ iterations of Lemma \ref{lem:threeGraphExtraction}, producing $k$ (3)-graphs $H_1 = (V,F_1),\dots,H_k = (V,F_k)$. Each iteration takes $O(\Delta \cdot \log n)$ rounds. Then, we edge color each subgraph $H_i$ with a fresh set of three colors using Lemma \ref{lem:ColorThreeGraphs} in $O(\log n)$ rounds. Finally, if $\Delta$ is even, the remaining graph $G_k$ has maximum degree at most two and can be 3-edge colored in $O(\log^* n)$ rounds. If $\Delta$ is odd, the final graph has maximum degree one and can trivially be edge colored with one color in a single round. In total the algorithm uses $\lfloor 3\Delta/2 \rfloor$ colors and runs in $O(\Delta^2 \cdot \log n)$ rounds.
\qed

\begin{algorithm}
    \caption{$3\Delta/2$-edge coloring}\label{alg:edge coloring}
    \begin{algorithmic}[1]
        \State $G_0 \gets G = (V,E)$
        \State $k \gets \lfloor \frac{\Delta - 1}{2}\rfloor$
        \For{$i = 1,2,\dots, k$}
            \State \Cref{lem:threeGraphExtraction}: Extract a (3)-graph $H_i = (V, F_i)$ from $G_{i-1}$. \Comment{$O(\Delta \cdot \log n)$ rounds}
            \State \Cref{lem:ColorThreeGraphs}: edge color $H_i$ with a fresh set of three colors. \Comment{$O(\log n)$ rounds}
            \State $G_i \gets (V, E_{i-1} - F_i)$ now has maximum degree at most $\Delta - 2i$.
        \EndFor
        \If{$\Delta$ is even}
        \State Color $G_k$ (maximum degree 2) with a fresh set of three colors \cite{linial92,cole86}. \Comment{$O(\log^* n)$ rounds}
        \Else
        \State Color $G_k$ (maximum degree 1) with a single fresh color.
        \EndIf
    \end{algorithmic}
\end{algorithm}

\subsection{Extracting (3)-graphs}
\label{sec:threeGraphExtraction}
In this section we apply the left side saturating matching algorithm under the condition $\delta>r$ as an improved subroutine for extracting a (3)-graph $G'=(V,E')\subseteq G$ in a way that reduces the maximum degree of the remaining graph by at least two in each iteration. To be self-contained we present the full algorithm and its proof, despite the large similarity to the slower approach in \cite{GKMU18}. 

First, we sketch our changes in their algorithm  for extracting a single (3)-graph; thereafter we present the whole algorithm. See \Cref{alg:threeGraphExtraction} for pseudocode. The extraction process is a combination of computing a sequence of maximal and maximum matchings in carefully chosen graphs. In essence, the union of the computed matchings will form the resulting (3)-graph. 
In this extraction procedure they construct bipartite graphs with minimum degree $\Delta$ on one side and maximum degree $\Delta - 1$ on the other side. The main bottleneck in their algorithm is the computation of a maximum matchings in these bipartite graphs. They use $O(\Delta\log n)$ iterations in each of which the current matching is augmented along a maximal independent sets of augmenting paths; the length of these paths is bounded by $\Theta(\Delta \log n)$. Since simulating augmenting paths as hyperedges introduces a communication overhead of $\ell$ rounds per iteration, this procedure intrinsically requires $\Omega(\Delta^2 \log^2 n)$ rounds, even if we completely disregard the complexity of computing the maximal independent sets. 

Instead, we use the distributed version of Hall's theorem, or more concretely \Cref{cor:maximumMatchingDeterministic}, to compute these matchings in  $O(\Delta \cdot \log n)$ rounds.

To present the whole extraction algorithm we require the following well-known results for computing maximal matchings. 
\begin{lemma}[Maximal Matching, deterministic, \cite{panconesi-rizzi}] \label{lem:maxmatch}
    There is a deterministic $O(\Delta + \log^* n))$-round
    algorithm that computes a maximal matching in graphs with maximum degree $\Delta$.
\end{lemma}

Let $T_{\mathrm{BM}}(n,\delta,r)$ to be the runtime of a bipartite maximum matching algorithm with $n$ nodes, maximum degree at most $\delta$ and rank at most $r$. Let $T_{\mathrm{M}}(n,\Delta)$ to be the runtime of a maximal matching algorithm on regular graphs with maximum degree $\Delta$. The framework of \cite{GKMU18} essentially yields the following result. As the runtime is parameterized entirely differently and to be self-contained we repeat the short algorithm and its proof. 

\begin{lemma}[Extracting (3)-graphs \cite{GKMU18}] \label{lemma:extracting-3-graphs}
    Let $G = (V,E)$ be a graph with maximum degree $\Delta \geq 3$. There is a deterministic distributed algorithm with time complexity
    $$
    O(
        T_{\mathrm{M}}(n,\Delta) + 
        T_{BM}(n,\Delta,\Delta - 1) +
        T_{\mathrm{M}}(n,\Delta - 1) + 
        T_{BM}(n,\Delta - 1,\Delta - 2) +
        T_{\mathrm{M}}(n,3)
    )
    $$
     that computes an edge set $F \subseteq E$ such that $H = (V,F)$ is a (3)-graph and the maximum degree of the graph $(V,E - F)$ is at most $\Delta - 2$.
\end{lemma}

\Cref{lem:threeGraphExtraction} follows from \Cref{lemma:extracting-3-graphs} as the runtime of the maximal matching terms is bounded by $O(\Delta+\log^*n)$ via \Cref{lem:maxmatch},  and the runtime of the  node saturating bipartite matching is upper bounded by $O(\log_{\frac{\Delta}{\Delta-1}}n)=O(\Delta\log n)$ via \Cref{cor:maximumMatchingDeterministic}. 

In the following for a graph $G=(V,E)$ and a set of edges $M\subseteq E$, we write $G-M\coloneqq(V,E\setminus M)$ for the graph $G$ without the edges in $M$. 

\begin{algorithm}
    \caption{$\mathsf{ReduceDegree}(G,\Delta)$: Returns a set of edges $M$ such that the maximum degree of the remaining graph $(V,E-M)$ is at most $\Delta - 1$.}
    \begin{algorithmic}[1]
        \State $V_\Delta \gets \{v \in S \mid \deg_G(v) = \Delta\}$
        \State  Compute a maximal matching $M_\Delta$ of $G[V_\Delta]$ (\Cref{lem:maxmatch}).
        \State $G' \gets (V, E - M_\Delta); \quad V_{\Delta}' \gets \{v \in S \mid \deg_{G'}(v) = \Delta\}$
        \State Let $B$ be the bipartite subgraph of $G'$ spanned by $V_{\Delta}'$ and $V - V_{\Delta}'$.
        \State Compute a maximum matching $M_B$ of $B$ (\Cref{cor:maximumMatchingDeterministic}).
        \State \Return $M_\Delta \cup M_B$
      \end{algorithmic}
      \label{alg:threeGraphExtraction}
    \end{algorithm}

\begin{algorithm}
\caption{$\mathsf{ThreeGraphExtraction}(G,\Delta)$: Returns a (3)-graph $H = (V,F)$ such that the maximum degree of the graph $(V, E - F)$ is at most $\Delta - 2$.}
\begin{algorithmic}[1]
    \State $M_1 \gets \text{ReduceDegree}(G,\Delta)$
    \State $M_2 \gets \text{ReduceDegree}(G - M_1, \Delta - 1)$
    \State Lemma \ref{lem:maxmatch}: Compute a maximal matching $M'$ of degree 3 nodes in $H'$.
    \State \Return $H = (V, M_1 \cup M_2 - M')$
  \end{algorithmic}
  \label{alg:reduceDegree}
\end{algorithm}

\medskip
\noindent
\textit{Proof of \Cref{lemma:extracting-3-graphs}.}\quad
See \Cref{alg:threeGraphExtraction,alg:reduceDegree} for pseudocode of the algorithm. 
    Let $V_\Delta$ denote the set of vertices of degree $\Delta$ in $G$. In the first step of the algorithm we compute a maximal matching $M_\Delta$ of $V_\Delta$ and define $G' := (V, E - M_\Delta)$. Next, we consider the set $V_{\Delta}'$ of unmatched vertices in $V_\Delta$. Notice that $V_{\Delta}'$ forms an independent set, as an edge between two vertices in $V_{\Delta}'$ would contradict the maximality of $M_\Delta$. Thus every vertex in $V_{\Delta}'$ has exactly $\Delta$ neighbors in the complementary set $V - V_{\Delta}'$. Further, since $V_{\Delta}'$ contains all vertices with degree $\Delta$ in $G'$, the bipartite subgraph $B$ of $G'$ spanned by $V_{\Delta}'$ and $V - V_{\Delta}'$ satisfies
    $$
        \Delta = \min_{u \in V_{\Delta}'} \deg_{B}(u) > \max_{v \in V - V_{\Delta}'} \deg_{B}(v) = \Delta - 1.
    $$
    Hence, we may apply Corollary 
    \ref{cor:maximumMatchingDeterministic} to obtain a maximum matching $M_B$ of $B$ that saturates all vertices in $V_{\Delta}'$. Then we repeat this procedure to reduce the maximal degree further down to $\Delta - 2$.
    Finally, in order not to remove more edges than necessary (and to obtain the degree $3$ nodes form an independent set), we compute a maximal matching $M'$ of the degree-3 nodes in $M$ and re-add them to $G$ at the end of the procedure.
    Now we will prove the following two central claims about this algorithm:

    \textbf{H is a valid (3)-graph:} 
    First, we observe that each execution of $\mathsf{ReduceDegree}$ can add at most two incident edges per vertex. Hence, it only remains to argue that a vertex $v$ with two incident edges in $M_1$ cannot gain two additional incident edges in $M_2$. Clearly, it holds that $\deg_{G - M_1}(v) \leq \Delta - 2$. Hence $v$ cannot participate in the maximal matching of vertices with degree $\Delta - 1$ in $G - M_1$ and can gain at most one additional incident edge during the second execution of $\mathsf{ReduceDegree}$. Thus the maximum degree of $H'$ is $3$. This implies that also $H\subseteq H'$ has maximum degree $3$. 
    
    Now, as $H$ is formed by removing a maximal matching between the degree $3$ nodes in $H'$, $H$ does not contain any two adjacent nodes with degree $3$. 

    \textbf{The maximum degree of $G\setminus H$ is at most $\Delta-2$:}
    We argue that each iteration of $\mathsf{ReduceDegree}$ reduces the maximal degree of $G$ by at least one. We observe that every vertex that is not matched by $M_\Delta$ must be matched by $M_B$, since our bipartite maximum mathching algorithm is guaranteed to saturate the $V_{\Delta}'$-side of the bipartite graph. Finally, adding the edges of $M'$ back cannot increase the maximal degree to above $\Delta - 2$, since both vertices have still degree two in $H$.
\qed

\subsection{Edge Coloring (3)-graphs with 3 Colors}
\label{sec:colorThreeGraphs}
In this section, we prove \Cref{lem:ColorThreeGraphs}, i.e., we show that $(3)$-graphs can be $3$-edge-colored in $O(\log n)$ rounds.
For an overview of our approach, we refer the reader to \Cref{sec:nutedge}.
Before stating the algorithm we will use to obtain \Cref{lem:ColorThreeGraphs}, we need a few definitions.
\begin{definition}
    Let $G=(V,E)$ be a graph and let $W \subseteq V$ and $F \subseteq E$. Then we denote
    \begin{itemize}
        \item by $G[W]$ the graph induced by nodes $W$ in $G$,
        
        \item by $G[F]$ the graph induced by edges $F$ in $G$, and
         
        \item by $N_G[W]$ the union of $W$ with the set of nodes that are adjacent to some node in $W$.
    \end{itemize}
    We may simply write $N[W]$ if there is no ambiguity for the choice of $G$, and set $N_G(v) := N_G[\{ v \}]$.
    Moreover, for a subgraph $G'$ of $G$, set $N[G'] := N[V(G')]$.
    Finally, the \emph{edge degree} of an edge $e$ is defined as the number of edges that are adjacent to $e$.
\end{definition}

\begin{definition}
    Let $G$ be a directed graph.
    For a directed edge $e = (u, v)$, we call $u$ the \emph{tail} of $e$, denoted by $\tail(e)$, and $v$ the head of $e$, denoted by $\head(e)$.
    Moreover, for two edges $e_{1} \neq e_{2}$ in $G$, we call
    \begin{itemize}
        \item $e_{1}$ a \emph{sibling} of $e_{2}$ if $\head(e_{1}) = \head(e_{2})$ and
        
        \item $e_{1}$ a \emph{child} of $e_{2}$ and $e_{2}$ a \emph{parent} of $e_{1}$ if $\tail(e_{2})= \head(e_{1})$.        
    \end{itemize}
\end{definition}

Now we are ready to state the $3$-edge-coloring algorithm $\mathcal A$ that we will use for the proof of \Cref{lem:ColorThreeGraphs}.
In the remainder of this section, we assume the input graph $G = (V,E)$ to be a $(3)$-graph.
Algorithm $\mathcal A$ proceeds in $3$ steps as follows.

\paragraph{Step 1: Partitioning the nodes into clusters.}
    In this step, the nodes are partitioned into clusters such that the subgraphs induced by each of these clusters are connected, pairwise disjoint, and of constant diameter.
    
    Compute a maximal independent set $\mathcal{I}$ on the power graph $G^{9}$, i.e., on the graph obtained from $G$ by adding (to the already existing edge set) an edge between any two nodes of distance between $2$ and $9$. Now each node $x$ chooses a node $d_{x}$ (which will be used to determine the cluster to which it belongs) using the following clustering process (that will ensure connectedness of the subgraphs induced by the clusters):
        \begin{algorithm}[H]
        \caption*{Clustering process}\label{clustering}
            \begin{algorithmic}
            \State $d_{x} \gets \phi$ , $i \gets 1$
            \If{$x \in \mathcal{I}$}
                \State $d_{x} \gets x$
                \State \Return 
            \EndIf
            \While{$i \leq 9$}
                \If{($d_{y} = p \neq \phi$ for some $y \in N_{G}(x)$)}
                    \State $d_{x} \gets p$ (breaking ties arbitrarily if the condition is satisfied for more than one $y \in N_{G}(x)$)
                    \State \Return 
                \EndIf 
                \State $i \gets i+1$
            \EndWhile
            \end{algorithmic}
        \end{algorithm}
    Each node $x$ gets a non-null value for $d_{x}$ by this clustering process.\\ 
    Define for each $i \in  \mathcal{I}$,
        \begin{align*}
        V_{i} &\coloneqq \{x \in V : d_{x} = i\} \text{, and}\\
        G_i &\coloneqq G[V_i] \text{,}
        \end{align*}
    and define $E_{i}$ as the set of edges of $G_{i}$.
    Note that, by construction, $G_{i}$ is connected, is disjoint from $G_{j}$ if $i \neq j$, and has constant diameter.
    Moreover, all nodes that are within distance $4$ from some node $i \in \mathcal I$ belong to $V_{i}$. We may abuse the term \emph{cluster} to refer to either $V_{i}$ or $G_i$.
        
    The described clustering induces a partitioning of the edges of $G$ into two sets $E_{\intra}$ and $E_{\inter}$ by defining
        \begin{align*}
        E_{\intra} & \coloneqq \bigcup_{i \in \mathcal{I}}E_{i}\\ 
        E_{\inter} & \coloneqq E\setminus E_{\intra}
        \end{align*} 
    We will refer to the edges in $E_{\intra}$ and $E_{\inter}$ as \emph{intracluster edges} and \emph{intercluster edges}, respectively. Since every node is in some cluster $G_{i}$ and any cluster is connected, every node is incident to at least one intracluster edge. This implies that $G[E_{\inter}]$ has maximum degree $2$, and combining this insight with the fact that a $(3)$-graph does not contain adjacent nodes of degree $3$ (or larger), we obtain the following observation.

    \begin{observation}\label{obs:lengthtwo}
        $G[E_{\inter}]$ is a union of disjoint paths of length at most $2$.
    \end{observation}

\paragraph{Step 2: Coloring the intercluster edges.}
    Call a cluster $G_{i}$ \emph{expanding} if it has at least $9$ adjacent intercluster edges.
    Consider the multigraph $H=(V(H),E(H))$ defined by 
        \begin{align*}
            V(H) & \coloneqq \{V_{i}:i \in \mathcal{I}\} \text{, and} \\
            E(H) & \coloneqq \{ (V_{d_x},V_{d_y}): \{x,y\}\text{ is an intercluster edge}\}.
        \end{align*}
    In other words, $H$ is the multigraph obtained from $G$ by contracting clusters.
    
    We start Step 2 by computing an orientation on $H$ such that each node with degree at least $9$ has at least $3$ outgoing edges as follows: Consider the graph $H'$ obtained from $H$ by splitting each node $v \in V(H)$ into three copies $v_1, v_2, v_3$ such that for each edge $e$ incident to $v$ in $H$, the endpoint $v$ of $e$ is replaced by precisely one of $v_1, v_2, v_3$. More precisely, we split the edges incident to $v$ as evenly as possible between these three nodes (as endpoints), i.e., the degrees of $v_j$ and $v_k$ in $H'$ differ by at most $1$, for any $j \neq k \in \{ 1, 2, 3 \}$.
    In particular, for each node $v \in V$ of degree at least $9$, we have $\deg(v_j) \geq 3$, for each $j \in \{ 1, 2, 3 \}$.
    Compute a sinkless orientation on $H'$ by using
    the deterministic sinkless orientation algorithm from~\cite[Corollary 4]{Splitting20}.    
    This provides an outgoing edge for each of the nodes of degree at least $3$ in $H'$ and therefore gives the desired orientation on $H$.
    
    Next, each node in $H$ with degree at least $9$ chooses $3$ of its outgoing edges. This naturally corresponds to a set of $3$ chosen edges in $E_{\inter}$ for each expanding cluster.
    Note that, by construction, any intercluster edge is chosen for at most one cluster.
    
    Now, compute a $3$-edge coloring $\varphi$ of $G[E_{\inter}]$ greedily (which can be done in constant time due to \Cref{obs:lengthtwo}) and go through the color classes of $\varphi$ in phases to compute a new coloring $\psi$. In phase $c$ corresponding to color (class) $c$, each edge $e$ of color $\varphi(e) = c$ receives a new color $\psi(e)$ as follows.

    If $e$ is not chosen for any cluster, then we assign to $e$ an arbitrary color that is distinct from the colors (in $\psi$) assigned to edges adjacent to $e$ that are already colored in $\psi$.
    We do the same if $e$ is chosen for some cluster $G_i$ that is not a tree.
    What remains is to define how to color any edge $e$ that is one of the three edges chosen for some cluster that is a tree.
    
    To this end, consider a cluster $G_i$ whose corresponding node in $H$ has degree at least $9$, and let $e_1$, $e_2$, and $e_3$ be the chosen edges of cluster $G_i$.
    For determining the color $\psi(e_1)$, $\psi(e_2)$, and $\psi(e_3)$, we proceed as follows.
    
    Let $G'_i$ denote the subgraph of $G$ induced by all edges that have at least one endpoint in $G_i$.
    Note that $G'_i$ is not necessarily a tree, but for each edge in $G'_i$ the two endpoints of the edge have a different distance to node $i$ (since $G_i$ is a tree). 
    Orient the edges in $G'_i$ towards $i$.
    W.l.o.g., assume that $\varphi(e_{1}) \leq \varphi(e_{2}) \leq \varphi(e_{3})$, i.e., we can assume that the edges are to be colored in the order $e_1, e_2, e_3$.
    (Note that if $\varphi(e_j) = \varphi(e_{j+1})$ for some $j \in \{ 1, 2\}$, we can still operate under this assumption by simply waiting with the decision how to color $e_{j+1}$ in $\psi$ until the color of $e_j$ is fixed.)
    Now color the edges $e_1, e_2, e_3$ as described in \Cref{algo:colproc}, which is based on an exhaustive case distinction depending on the degrees of the heads of $e_1$, $e_2$ and $e_3$ (each of which must be either $2$ or $3$ as the head of either of these edges is $\neq i$).
    Moreover, we call a color \emph{available} for an edge if the color has not already been assigned to some adjacent edge in an earlier phase. 
        \begin{algorithm}[h]
        \caption{Coloring Procedure} \label{algo:colproc} 
        \begin{outline}
        \1  \emph{Case 1: There exist distinct $k,\ell \in \{1,2,3\}$ such that the degree of both $head(e_k)$ and $\head(e_{\ell})$ is $2$.}\\
            Color $e_k$ and $e_l$ with different available colors. (This is possible since each of $e_k$ and $e_l$ is adjacent to at most one intercluster edge, by \Cref{obs:lengthtwo}.)
        \1  \emph{Case 2: For exactly one chosen edge $e_k$, the degree of $\head(e_k)$ is $2$.}
            \2  \emph{Subcase 1: $k \in \{1,2\}$}\\
                Color $e_1$ and $e_2$ with available colors. Color $e_3$ with the color of $e_k$ if that color is available for $e_3$; otherwise color $e_3$ with an arbitrary available color.
            \2  \emph{Subcase 2: $k=3$}\\
                Color $e_1$ and $e_2$ with different available colors. Color $e_3$ with one of the colors of $e_1$ or $e_2$ that is available for $e_3$. (This is possible since at most one intercluster edge is adjacent to $e_3$.)
        \1  \emph{Case 3: $\head(e_k)$ does not have degree $2$ for any $k\in \{1,2,3\}$}.\\
             Color $e_1$ and $e_2$ with different available colors. Color $e_3$ with a color that is different from the colors of both $e_1$ and $e_2$, if such a color is available for $e_3$. If no such color is available, then color $e_3$ with an arbitrary available color.
        \end{outline}
        \end{algorithm}

\paragraph{Step 3: Extending the coloring to intracluster edges.}
    In this step, we extend the coloring of the intercluster edges to a proper 3-edge coloring of the entire graph $G$ by coloring the intracluster edges (without modifying the colors of the intercluster edges).
    To this end, each node $i \in \mathcal{I}$ simply collects $G'_i$ and then chooses one of the colorings of the intracluster edges in $G_i$ that does not conflict with the coloring of the intercluster edges adjacent to $G_i$.
    As we will show later in \Cref{lem:extension-col}, the coloring $\psi$ of the intercluster edges guarantees that such a non-conflicting coloring of the intracluster edges exists.
        
    This concludes the description of Algorithm $\mathcal A$.
    Next, we bound the runtime of algorithm $\mathcal A$.

\begin{lemma}\label{lem:threecolorrun}
    Algorithm $\mathcal A$ can be performed in $O(\log n)$ rounds.
\end{lemma}
\begin{proof}
    We start by observing that algorithms executed on $G^9$, $H$, or $H'$  can be simulated on $G$ with only a constant-factor overhead, due to the fact that the considered clusters are of constant diameter (which implies that one round of communication on any of the aforementioned graphs can be simulated in a constant number of rounds on $G$).
    Therefore, we can treat any of the aforementioned graphs as the underlying communication graph when necessary, without incurring any asymptotic change in the overall runtime.
    Now we bound the complexity of the different (nontrivial) steps of Algorithm $\mathcal A$ one by one.
    
    Given the above observation and the fact that the maximum degree of $G$ is $3$ (which implies that the maximum degree of $G^9$ is constant), we can compute a maximal independent set of $G^9$ in $O(\log^* n)$ rounds by using the $O(\log^* n + \Delta)$-round algorithm from~\cite{BEK14}.
    The subsequent clustering process takes a constant number of rounds.
    Hence, \emph{Step 1} of algorithm $\mathcal A$ can be performed in $O(\log^* n)$ rounds.
    
    Computing a sinkless orientation as described takes $O(\log n)$ rounds, due to the same bound in~\cite[Corollary 4]{Splitting20}.
    As already observed in the algorithm description, $\varphi$ can be computed in a constant number of rounds.
    For each of the three color classes (in $\varphi$), the new colors (in $\psi$) of the respectively considered edges can be computed in a constant number of rounds as well: for the chosen edges of a cluster $G_i$, we can simply assume that node $i$ computes the new colors and sends them to the respective edges; as the new colors only depend on information in $G'_i$ (which has constant diameter), this takes a constant number of rounds.
    Hence, \emph{Step 2} can be performed in $O(\log n)$ rounds.
    
    Extending the coloring to the intracluster edges takes a constant number of rounds by design, and hence \emph{Step 3} can be performed in a constant number of rounds. Therefore, the overall runtime of Algorithm $\mathcal A$ is $O(\log n)$ rounds.
\end{proof}

The only remaining ingredient we need for the proof of \Cref{lem:ColorThreeGraphs} is to show that the coloring of the intercluster edges computed in Step 2 of $\mathcal A$ can always be completed to a proper coloring on the entire graph.
To this end, it suffices to show extendability of the coloring for each cluster individually (as the clusters are vertex-disjoint), which we will take care of in \Cref{lem:extension-col}.
Before stating and proving \Cref{lem:extension-col}, we need one more definition.

\begin{definition}
    Consider a graph $G'_i$ as defined in Algorithm $\mathcal A$ and let $k$ be a non-negative integer. Define
    \begin{align*}
        L_{k} \coloneqq \{e \in E: e \text{ is at distance $k$ from $i$}\},
    \end{align*}
    where an edge incident to $i$ is assumed to have distance $0$ from $i$. 
    We call $L_{v}$ the \emph{layer at distance $k$} from $i$.
    Moreover, for a partial coloring of the edges of $G'_i$, we call the set of colors that appear in layer $L_{k}$ the \emph{color palette of $L_{k}$}.
\end{definition}
    
We are now ready to prove \Cref{lem:extension-col}.

\begin{lemma}\label{lem:extension-col}
    Let $G_i$ be a cluster. Then, the $3$-edge coloring $\psi$ of $G[E_{\inter}]$ can be (properly) extended to $G_i$.
\end{lemma}

\begin{proof}
    We start by observing that every edge in $G$ has edge degree at most $3$ (as $G$ is a $(3)$-graph).
    Now consider the following exhaustive cases:
    \begin{enumerate}
    \item \emph{$G_{i}$ contains an edge $e$ of edge degree at most $2$ (in $G$).}\\
        Extend the $3$-edge coloring of the intercluster edges to a $3$-edge coloring of $G_{i}$ by coloring the edges of $G_{i}$ in an order of non-increasing distance from $e$, i.e., color an edge only after all the edges that are farther from $e$ are colored. This can be done greedily with $3$ colors since for all edges except $e$, there is always an adjacent edge that is yet to be colored (and each edge has edge degree at most $3$ as observed above). We can choose a color for $e$ irrespective of the color of its adjacent edges, since its edge degree is $2$.
    \item \emph{$G_{i}$ contains a cycle $C$.}\\
        If $C$ is odd, it must be that two nodes of degree $2$ are adjacent since degree-$3$ nodes cannot be adjacent as $G$ is a $(3)$-graph. It follows that the edge between those two nodes of degree $2$ has edge degree $2$, which reduces to an instance of the previous case (and we are done). Hence, assume that $C$ is even. Then greedily $3$-color the edges of $G_{i}$ that do not lie on $C$ in an order of non-increasing distance from $C$. This leaves us with a $2$-list edge coloring problem for an even cycle, which is known to always have a solution (see, e.g., \cite[Section 2.1]{GhaffariHKM21}). 
    \item \emph{$G_{i}$ is a tree and each edge of $G_{i}$ has edge degree $3$ (in $G$).}\\
        Note that in this case, the fact that each node within distance $4$ of $i$ is contained in $G_i$ ensures that the node in $H$ corresponding to $G_i$ has degree at least $9$.
        Therefore, there will be three edges chosen for $G_i$ whose colors will be determined by \Cref{algo:colproc}.
            Consider the subgraph $G'_i$ along with the partial coloring it obtained from Step 2 of Algorithm $\mathcal A$ and recall that the edges in $G'_i$ are oriented towards $i$.
            
        Greedily $3$-color the edges of $G_i$ in an order of non-increasing distance from $i$ (until this is not possible anymore), i.e., color an edge only after all the edges in a farther layer from $i$ have been colored. This will color all the edges in $G_{i}$ except potentially an edge $f$ adjacent to $i$ since for any edge that is not colored last, there is always an adjacent edge yet to be colored, which guarantees that at most two of its adjacent edges are already colored. If there is no such edge $f$, we are done, hence assume that $f$ exists (and is uncolored).
        Call an edge $e$ \emph{friendly} if one of $e$'s children is colored the same as one of $e$'s siblings.

        We claim that at least one of the edges of $G_i$ is friendly (w.r.t.\ the current coloring).
        For a contradiction, assume that $G_i$ does not contain any friendly edge.
        Note that the fact that each edge of $G_{i}$ has edge degree $3$ implies that for any edge in $G_i$, one of the two endpoints has degree $2$ and the other degree $3$, which implies that one of the following must be true.
        \begin{enumerate}
            \item\label{item:twofirst} For each edge $e$ in $G'_i$, $\head(e)$ has degree $2$ if $e \in L_{s}$ for some even $s$ and degree $3$ if $e \in L_{s}$ for some odd $s$.
            \item\label{item:threefirst} For each edge $e$ in $G'_i$, $\head(e)$ has degree $3$ if $e \in L_{s}$ for some even $s$ and degree $2$ if $e \in L_{s}$ for some odd $s$.
        \end{enumerate}

        In the first case, using the assumption that $G_i$ does not contain any friendly edge, it is straightforward to show by induction that if $c \in \{ 1, 2, 3 \}$ is the color of the unique colored edge incident to $i$, then, for any even $s$, the color palette of (any nonempty) $L_s$ is $\{ c \}$ and, for any odd $s$ the color palette of (any nonempty) $L_s$ is $\{ 1, 2, 3 \} \setminus \{ c \}$.
        Analogously, in the second case, it is straightforward to show by induction that if $c' \neq c''$ are the colors of the two colored edge incident to $i$, then, for any even $s$, the color palette of (any nonempty) $L_s$ is $\{ c', c'' \}$ and, for any odd $s$ the color palette of (any nonempty) $L_s$ is $\{ 1, 2, 3 \} \setminus \{ c', c'' \}$.
        Hence, in either case, there is a color $c$ such that all edges of $G'_i$ whose head has degree $2$ have color $c$ and all edges of $G'_i$ whose head has degree $3$ have a color from $\{1,2,3\} \setminus \{ c \}$.
        We now show that this yields a contradiction for either of the three cases in \Cref{algo:colproc}.
        \begin{itemize}
            \item In case 1 of \Cref{algo:colproc}, we obtain a contradiction due to the fact that there are two differently colored edges whose head has degree $2$.
            \item In case 2, subcase 1, either $e_3$ and $e_k$ have the same color or $e_3$ must have a sibling that has the same color as $e_k$ (as the tail of $e_3$ has no incident intercluster edge except for $e_3$ itself).
            In either case, we obtain a contradiction due to the fact that there are two same-colored edges for one of which the degree of its head is $2$ while for the head of the other the degree is $3$.
            \item In case 2, subcase 2, we again obtain a contradiction due to the fact that there are two same-colored edges for one of which the degree of its head is $2$ while for the head of the other the degree is $3$.
            \item In case 3, similarly to case 2, subcase 1, either $e_1$, $e_2$, and $e_3$ have pairwise different colors or $e_3$ must have a sibling $e'_3$ such that $e_1$, $e_2$, and $e'_3$ have pairwise different colors.
            In either case, we obtain a contradiction.
        \end{itemize}
        This implies that our assumption was false and proves the claim that at least one of the edges of $G_i$ is friendly.

        Let $e$ be such a friendly edge, and let $e'$ and $e''$, respectively, denote a sibling and a child of $e$ that have the same color. 
        Consider the unique path $P$ in $G_i$ starting at edge $e$ and ending at edge $f$.
        If we ignore $f$, then $P$ is directed, which implies that both of $e'$ and $e''$ do not lie on $P$.
        Now, uncolor all colored edges on $P$ and then greedily color all edges on $P$ from $f$ towards $e$.
        As before, this is possible for all edges $\neq e$ due to the fact that each such edge has an uncolored adjacent edge at the time it is colored.
        Moreover, $e$ has two adjacent edges of the same color, which (together with the fact that $e$ has edge degree at most $3$) implies that there is also an available color for $e$.
        We conclude that the $3$-edge coloring $\psi$ of $G[E_{\inter}]$ can be extended to $G_i$, as desired. \qed
    \end{enumerate}      

\renewcommand{\qed}{}
\end{proof}

By combining the insights from this section we obtain \Cref{lem:ColorThreeGraphs}.

\medskip
\noindent
\textit{Proof of \Cref{lem:ColorThreeGraphs}.}\quad
By \Cref{lem:extension-col}, Algorithm $\mathcal A$ computes a proper $3$-coloring of any $(3)$-graph $G$.
By \Cref{lem:threecolorrun}, the runtime of Algorithm $\mathcal A$ is $O(\log n)$ rounds. \qed

\subsection{A Faster Algorithm for $(3/2+\eps)\Delta$-Edge Coloring}
\label{sec:edgeColoringProof}

At the expense of an arbitrarily small amount of additional colors, one can further reduce the runtime of \Cref{thm:edgeColoring}. This section is devoted to proving the following corollary. 
\corEdgeColoring*

For its proof, we require the following theorem. 

\begin{theorem}[{\cite[Theorem 1]{Splitting20}}] \label{thm:splitting}
    For every $\gamma > 0$, there are deterministic $O(\gamma^{-1} \cdot \log \gamma^{-1} \cdot (\log \log \gamma^{-1})^{1.71} \cdot \log n)$-round distributed algorithms for computing undirected degree splittings such that the discrepancy at each node $v$ of degree $d(v)$ is at most $\gamma \cdot d(v) + 4$.
\end{theorem}

This result was originally used to compute a $(2 + \eps)\Delta$-edge coloring \cite[Corollary 1]{Splitting20} by recursively splitting the graph into smaller parts that are then colored with a classic $(2\Delta-1)$-edge coloring algorithm. We will now closely follow this argument, but use our new $3\Delta/2$-edge coloring algorithm instead of a standard $(2\Delta - 1)$-edge coloring algorithm for the base case.

\medskip
\noindent
\textit{Proof of \Cref{cor:EdgeColoringFast}.}\quad
We start by applying Theorem \ref{thm:splitting} with parameter $\gamma = \frac{\eps}{20 \log \Delta}$ for $h = \log \frac{\eps \Delta}{15}$ iterations and each of the parts in parallel.
Let $\Delta_{i-1}$ denote the maximum degree of each part before iteration $i$. 
Then, one finds that $\Delta_{i} \leq \frac{1}{2}(\Delta_{i-1} + \gamma\Delta_{i-1} + 4)$ and further via induction on the number of iterations
$$
\Delta_{i} \leq \left( \frac{1+\gamma}{2} \right)^{i}\Delta + 2\sum_{k=0}^{i-1} \left( \frac{1+\gamma}{2} \right)^{k.} \leq \left( \frac{1+\gamma}{2} \right)^{i}\Delta + 5.
$$

In the last step, we have used the geometric sum formula together with the estimate $\gamma \leq \frac{1}{10}$. Hence, after the final iterations we are left with $2^{h}$ subgraphs of maximum degree at most 
\begin{align*}\Delta_{h} =2^{-h}(1+\gamma)^h\Delta+5\leq 2^{-h} \cdot e^{\gamma\cdot h}\Delta+5=15/\eps \cdot e^{\gamma\cdot h} +5 =O\left( 1/\varepsilon \right).
\end{align*}
Now we use \Cref{thm:edgeColoring} to compute a $(3\Delta_h/2)$-edge coloring for each of these subgraphs in parallel, all with different sets of colors. Thus, we get an edge coloring of the whole graph with

\begin{align*}
2^{h} \cdot (3\Delta_h/2) &\leq \frac{3}{2} \cdot 2^{h}\left( \left( \frac{1+\gamma}{2} \right)^{h}\Delta + 5 \right) 
\leq \frac{3}{2}\Delta(1 + \gamma)^{\log \Delta} + \frac{\varepsilon}{2}\Delta \\
&= \frac{3}{2}\Delta\left( 1 + \frac{\varepsilon}{20\log\Delta} \right)^{\log\Delta} + \frac{\varepsilon}{2}\Delta 
\leq \frac{3}{2}\Delta \exp\left( \frac{\varepsilon}{20} \right) + \frac{\varepsilon}{2}\Delta \leq \left( \frac{3}{2} + \varepsilon \right)\Delta.
\end{align*}

Using $\eps>1/\Delta$, we can bound the runtime of each recursive split by
\begin{align*}
O\left( \frac{1}{\gamma}\cdot \log \frac{1}{\gamma} \cdot \log^{1.71} \log \frac{1}{\gamma} \cdot \log n \right)&= O\left(  \frac{\log\Delta}{\varepsilon} \cdot \log \frac{\log \Delta}{\varepsilon} \cdot \log^{1.71}\log \frac{\log\Delta}{\varepsilon} \cdot \log n \right) \\
&= O\left( \frac{\log \Delta}{\varepsilon} \cdot \log \log \Delta \cdot \log^{1.71}\log\log\Delta \cdot \log n \right) \\
& O\left(\eps^{-1}\cdot \log\Delta \cdot \log^2\log\Delta\cdot \log n\right).
\end{align*}
The coloring of the individual parts takes $O(\Delta_h^2 \cdot \log n)=O(\eps^{-2}\cdot\log n)$. 
Thus, total round complexity can be upper bounded by 
\begin{align*}
O(\eps^{-2}\cdot \log n+h\cdot \eps^{-1}\cdot\log\Delta\cdot \log^2\log \Delta\cdot \log n)=O(\eps^{-1}\cdot \log n \cdot \log^2 \Delta \cdot(\eps^{-1} +\log^2\log \Delta)). \qed
\end{align*} 

\section{Linear Lower Bound for HSO}
\label{sec:lowerbound}
\thmHSOLowerBoundGlobal*
\begin{proof}
    For any $\delta$, we provide a simple construction of an infinite hypergraph class $\mathcal G$ for which any $G \in \mathcal G$ has minimum degree and maximum rank $\delta = r$ and on which any deterministic HSO algorithm requires $\Omega(n)$ rounds.
    For simplicity, we will describe the hypergraphs $G \in \mathcal G$ via their bipartite representation $\cB_G$ (for which we will therefore assume that the number of nodes is $2n$).
    In the following, we describe the construction of the graphs $\cB_G$.
    
    Let $n$ be any positive integer such that $n - 1$ is a (positive) multiple of $\delta^2$.
    Let $H$ be the graph obtained from the complete bipartite graph $K_{\delta, \delta}$ by removing an edge $\{ u, v \}$.
    Consider the (not necessarily bipartite) directed graph $H' = (V',E')$ with node set $$
    V' :=\{ a, b\} \cup \{ w_{i,j} \mid 1 \leq i \leq \delta, 1 \leq j \leq (n - 1)/\delta^2 \},
    $$
    and edge set 
    $$
    E' := \{ ( a, w_{i,1} ) \mid 1 \leq i \leq \delta \} \cup \{ ( w_{i,j}, w_{i,j + 1} ) \mid 1 \leq i \leq \delta, 1 \leq j \leq (n - 1)/\delta^2 - 1 \} \cup \{ ( w_{i,(n - 1)/\delta^2}, b ) \mid 1 \leq i \leq \delta \}.
    $$
    Now, to obtain $\cB_G$ from $H'$, replace each $w_{i,j}$ by a copy of $H$ and replace each edge $(x,y)$ in $H'$ according to the following rules: if $x = a$, replace $(x,y)$ by an (undirected) edge between $a$ and node $u$ in the copy of $H$ corresponding to $y$; if $x \neq a$ and $y \neq b$, replace $(x,y)$ by an (undirected) edge between node $v$ in the copy of $H$ corresponding to $x$ and node $u$ in the copy of $H$ corresponding to $y$; if $y = b$, replace $(x,y)$ by an (undirected) edge between $b$ and node $v$ in the copy of $H$ corresponding to $x$.
    Our graph class $\mathcal G$ consists of precisely those graphs $G$, for which the bipartite representation $\cB_G$ can be obtained in the described manner (for some $n$ satisfying the stated property).

    From the construction, it follows directly that $\cB_G$ is a regular graph with minimum and maximum degree $\delta$.
    Moreover, for any perfect matching in $\cB_G$ the following must hold: if $w_{i,1}$ denotes the node in $H'$ corresponding to the copy of $H$ containing the matching partner of $a$ and $w_{i',(n - 1)/\delta^2}$ denotes the node in $H'$ corresponding to the copy of $H$ containing the matching partner of $b$, then $i = i'$.
    The reason for this is that otherwise the removal of $a$, $b$, and their matching partners from $\cB_G$ would make the remaining graph contain a maximal connected component with an odd number of nodes (namely, all remaining nodes in the copies of $H$ corresponding to the nodes $w_{i, j}$ for $1 \leq j \leq (n - 1)/\delta^2$), which is impossible as the considered matching is perfect.

    Now assume for a contradiction that there is an algorithm $\mathcal A$ solving the perfect matching problem in time $o(n)$ on the bipartite representations of all graphs in $\mathcal G$.
    Note that, by construction, the distance between $a$ and $b$ is in $\Omega(n)$.
    Hence, for sufficiently large $n$, the views of both $a$ and $b$ when executing Algorithm $\mathcal A$ on $\cB_G$ do not overlap.
    Consider an arbitrary assignment of IDs to the nodes of $\cB_G$, and consider the previously discussed indices $i, i'$ corresponding to the matching partners assigned to $a$ and $b$ by Algorithm $\mathcal A$.
    If, for the considered ID assignment, $i \neq i'$, we know that $\mathcal A$ is incorrect; hence assume that $i = i'$.
    Now, due to the symmetry of $\cB_G$, it is straightforward to rearrange the IDs in $a$'s view such that, executed on the new ID assignment, $\mathcal A$ will match $a$ with a different neighbor than before.
    As the views of $a$ and $b$ in $\mathcal A$ do not overlap, $\mathcal A$ will still match $b$ to the same neighbor, implying that $i \neq i'$, which yields a contradiction.
    Hence, there is no $o(n)$-round algorithm for perfect matching on the bipartite representations of all graphs in $\mathcal G$.

    As the regularity of $\cB_G$ implies that any matching saturating one side of the bipartition must be a perfect matching, we obtain that there is no $o(n)$-round algorithm for computing a matching saturating any side of the bipartition.
    By the equivalence between a correct solution for a matching saturating one side of the bipartition in a bipartite representation and a correct solution for HSO in the original hypergraph, we obtain the desired lower bound of $\Omega(n)$ for HSO.
\end{proof}

\bibliographystyle{ACM-Reference-Format}
\bibliography{references}

\end{document}